\documentclass[12pt,reqno]{amsart}
\usepackage[notref,notcite]{}
\usepackage{graphicx}
\usepackage{amssymb}
\usepackage{amsmath}

\newtheorem{theorem}{Theorem}
\newtheorem{lemma}{Lemma}
\newtheorem{proposition}{Proposition}
\newtheorem{problem}{Problem}

\begin{document}

\title{\bf Steady flows of ideal incompressible fluid}

\author[V. Yu. Rovenski and V. A. Sharafutdinov]{Vladimir Yu. Rovenski and Vladimir A. Sharafutdinov}


\address{Sobolev Institute of Mathematics; 4 Koptyug Avenue, Novosibirsk, 630090, Russia.}
\email{sharaf@math.nsc.ru}

\address{Department of Mathematics, University of Haifa,
3498838 Haifa, Israel.}
\email{vrovenski@univ.haifa.ac.il}

\begin{abstract}
{A new important relation between fluid mechanics and differential geometry is established.}
We study smooth {steady} solutions to the Euler equations with {the} additional property:
the velocity vector is orthogonal to the gradient of the pressure at any point. {Such solutions are called Gavrilov flows.}
Local structure of a Gavrilov flow {is} described in terms of geometry of isobaric hypersurfaces.
In the 3D case, we obtain a system of PDEs for an axisymmetric Gavrilov flow and find consistency conditions for the system.
Two numerical examples of axisymmetric Gavrilov flows are presented:
with pressure function periodic in {the} axial direction, and
with isobaric surfaces diffeomorphic to {the} torus.


\end{abstract}

\maketitle

\section{Introduction}

In dimensions 2 and 3, the Euler equations
\begin{eqnarray}
 &u\cdot\nabla u+{\rm grad}\,p=0,
                          \label{1.1}\\
 &\nabla\cdot u=0
                          \label{1.2}
\end{eqnarray}
describe steady flows of ideal incompressible fluid. The equations are also of some mathematical interest in an arbitrary dimension. Here $u=\big(u_1(x),\dots,u_n(x)\big)$ is a vector field on an open set $U\subset{\mathbb R}^n$ (the fluid velocity) and $p$ is a scalar function on $U$ (the pressure).
We consider only smooth real solutions to the Euler equations, i.e., $u_i\in C^\infty(U)\ (i=1,\dots,n)$ and $p\in C^\infty(U)$ are assumed to be real functions (the term ``smooth" is used as the synonym of ``$C^\infty$").
We say that a solution $(u,p)$ to \eqref{1.1}--\eqref{1.2} is a {\it {Gavrilov flow}} if it satisfies
\begin{equation}
u\cdot{\rm grad}\,p=0,
                          \label{1.3}
\end{equation}
i.e., the velocity is orthogonal to the pressure gradient at all points of $U$. {We use also the abbreviation GF for ``Gavrilov flow''.}

Equations \eqref{1.1}--\eqref{1.3} constitute {the} overdetermined system of first order differential equations: $n+2$ equations in $n+1$ unknown function. Therefore every GF is an exception in a certain sense. Nevertheless, such flows exist and deserve study.
Such flows satisfy the following important property: a pair of functions $(\tilde u,\tilde p)$ given~by
\begin{equation}
\tilde u=\varphi(p)u,\quad {{\rm grad}}\,\tilde p=\varphi^2(p){\rm grad}\,p,
                          \label{1.4}
\end{equation}
where $\varphi(p)$ is an arbitrary smooth function, is again a GF.
This property underlies the following construction that will be called the {\it Gavrilov localization}.
{Given a GF in a domain $U$,}
let $p_0\in U$ be a regular value of the function $p$ such that $M_{p_0}=\{x\in U\mid p(x)=p_0\}$ is a compact hypersurface in $U$.
Then we can construct a compactly supported smooth solution to the Euler equations on the whole of
${\mathbb R}^n$ by choosing $\varphi(p)$ as a cutoff function with support in a small neighborhood of $p_0$.
Indeed, the new velocity vector field $\tilde u$ and the gradient ${{\rm grad}}\,\tilde p$ are supported in some compact neighborhood $\tilde U\subset U$ of the surface $M_{p_0}$,
as {is} seen from \eqref{1.4}, and we define $\tilde u$ as zero in ${\mathbb R}^n\setminus U$.
Thus, the new pressure $\tilde p$ is constant on every connected component of $U\setminus \tilde U$.
Since only the gradient ${\rm grad}\,p$ participates in \eqref{1.1}--\eqref{1.3}, we can  assume without lost of generality that $\tilde p=0$ on the ``exterior component" of $U\setminus \tilde U$. It is now clear that $\tilde p$ can be extended to a compactly supported function $\tilde p\in C^\infty({\mathbb R}^n)$.

For some neighborhood $O({\mathcal C})$ of the circle ${\mathcal C}=\{(x_1,x_2,0)\in{\mathbb R}^3\mid x_1^2+x_2^2=R^2\}$, Gavrilov \cite{G} proved the existence of a solution $u\in C^\infty\big(O({\mathcal C})\setminus{\mathcal C};{\mathbb R}^3\big),\ p\in C^\infty\big(O({\mathcal C})\setminus{\mathcal C}\big)$ of the Euler equations satisfying \eqref{1.3}, and such that for some regular value $p_0$ of the function $p$, the surface $M_{p_0}\subset O({\mathcal C})\setminus{\mathcal C}$ is diffeomorphic to the torus.
Gavrilov's formulas involve an arbitrary positive constant $R$, without lost of generality we set $R=1$.
Using the localization procedure described above, Gavrilov proved the existence of a solution
$\tilde u\in C^\infty({\mathbb R}^3;{\mathbb R}^3),\ \tilde p\in C^\infty({\mathbb R}^3)$ of the Euler equations supported in a small neighborhood of $M_{p_0}$.
Thus, Gavrilov gave a positive answer to the {long standing} question:
Is there a smooth compactly supported solution to the Euler equations on ${\mathbb R}^3$ that is not identically equal to zero?
Unfortunately, \cite{G} is written in terse language and many
{technical}
details are omitted.
Actually the same idea is implemented in {the} subsequent article \cite{CLV} by Constantin -- La -- Vicol.
The latter paper starts with the so called Grad -- Shafranov ansatz that has appeared in plasma physics.
Unlike \cite{G}, the article \cite{CLV} involves a thorough analysis of nonlinear ODEs that arise while constructing a solution.

We emphasize that in both papers \cite{CLV,G} the existence of a Gavrilov {\it axisymmetric} smooth compactly supported flow on ${\mathbb R}^3$ is proved {only}.
Indeed, \cite{G} starts with the Euler equations in cylindrical coordinates for axisymmetric solutions.
To authors' knowledge,
The Grad -- Shafranov ansatz is adapted to the study of axisymmetric solutions.

Existence (or non existence) of compactly supported steady flows of another kind (i.e., not GFs) is discussed in \cite{CC-new,JX-new,N-new, NV}.

We~say that two GFs $(u,p)$ and $(\tilde u,\tilde p)$, defined on the same open set $U\subset{\mathbb R}^n$,
are {\it equivalent} if \eqref{1.4} holds with a smooth non-vanishing function $\varphi(p)$.
For example, $(u,p)$ and $(-u,p)$ are equivalent GFs. In the present article, GFs are considered up to the equivalence.
We mostly study the structure of such a flow in a neighborhood of the hypersurface $M_{p_0}\subset U$ for a regular value $p_0$ of the pressure.

\vskip2mm

The article is organized as follows. In Section~\ref{sec:geom}, {we present a complete description of a GF in terms of first and second quadratic forms of isobaric hypersurfaces $M_p$. After the description is obtained, the Euler equations can be forgotten.

A.V. Gavrilov was so kind as told his results to the second author and some other colleagues before the article \cite{G} was published. While discussing Gavrilov's results, Ya.V.~Ba\-zaykin suggested an explicit example of a GF in any even dimension (private communication, 2018). The example was independently reproduced in the recent paper by A.~Enciso, D.~Peralta-Salas, and F.~Torres de Lizaur \cite[Proposition 2.1]{EP-STdeL-new}. The example is discussed at the end of Section~\ref{sec:geom}.

The following observation is widely used in \cite{NV}. If a solution $(u,p)$ to the Euler equations is defined on the whole of ${\mathbb R}^3$ and sufficiently fast decays at infinity, then the quadratic form $(\nu\cdot u)(\xi\cdot u)$ integrates to zero over every 2D affine plane $P\subset{\mathbb R}^3$, where $\nu$ is the normal vector to $P$ and $\xi$ is an arbitrary vector parallel to $P$. For a GF, the two-dimensional integral over $P$ can be reduced to a one-dimensional integral over the curve $P\cap M_p$. The reduction is presented in Section 3. This property of GFs is interesting by itself, but so far we do not know any application of the property. Therefore Section 3 can be skipped on first reading.

In Section 4, we obtain a system of PDEs for GFs in three dimensions. It is an overdetermined system: 4 equations in 3 unknown functions. The problem of deriving consistency conditions for the system is the main problem in the study of GFs. The problem remains open in the general case.}
%
%

Sections~\ref{sec:axi}--\ref{sec:structure} are devoted to axisymmetric GFs in {the 3D case}. Unlike \cite{CLV,G}, our analy\-sis of
{such} flows is based on the well-known geometric fact: the equation for geodesics admits a first integral for surfaces of revolution,
the Clairaut integral.
In~Section~\ref{sec:axi},
we reduce the system {of} Section~\ref{sec:eqs} to a simpler system of PDEs for axisymmetric GFs: two equations in one unknown function of two variables, including also two functions of one variable $p$.
{Consistency conditions for the latter system} are derived in Section~\ref{sec:cond}.

In Section 7, we discuss a numerical method of constructing axisymmetric GFs. In~the general case, our method gives only a local GF in a neighborhood of a given point. Nevertheless, global  axisymmetric GFs can be found due to periodicity in the $z$-direction.

The existence of an axisymmetric GF in the open set $O({\mathcal C})\setminus{\mathcal C}$ is proved in \cite{G} such that the pressure function $p(r,z)$ is smooth in a neighborhood of the point $(r,z)=(1,0)$ that is a nondegenerate minimum point of $p$.
We~study such flows in Section~\ref{sec:structure}. The~corresponding system of PDEs can be solved in series.

We emphasize that \cite{CLV,G} only prove the existence of axisymmetric GFs but do not give numerical examples.
In our opinion, numerical and geometric examples are of a great importance since they can lead to new hypotheses.
In Sections~\ref{sec:2ex}--\ref{sec:structure}, we {present} two geometric illustrations {for} better understanding axisymmetric
GF{s}: one with isobaric surfaces diffeomorphic to a torus and {second one periodic in the $z$-direction}.

Some open questions on GFs are posed  in Section~\ref{sec:problems}.

\section{Geometry of {a Gavrilov flow}}
\label{sec:geom}

Let $(u,p)$ be a GF on an open set $U\subset{\mathbb R}^n$.
Integral curves of the vector field $u$ are also called \textit{particle trajectories}.
The following {statement immediately follows from \eqref{1.3}}.

\begin{proposition}\label{P2.1}
The pressure $p$ is constant on every integral curve of the vector field~$u$.
\end{proposition}


By Proposition~\ref{P2.1}, the Bernoulli law
{$$
|u|^2/2+p={\rm{const}}\quad\mbox{along a particle trajectory}
$$
splits, for a GF, into two conservation laws:}
\begin{equation}
p=c={\rm{const}},\quad {\vert u\vert}=C={\rm{const}}\quad\mbox{along a particle trajectory}.
                          \label{2.1}
\end{equation}

We say that $x\in U$ is a {\it regular point} if
${\rm grad}\,p(x)\neq0$. The vector field $u$ does not vanish at regular points as is seen from \eqref{1.1}. The sets
$
 M_{p_0}=\{x\in U\mid p(x)=p_0={\rm {const}}\}
$
will be called {\it isobaric hypersurfaces}.
A particle trajectory starting at a point of an isobaric hypersurface $M_p$ does not leave $M_p$ ``for ever''.
In the general case, an arbitrary closed subset of $U$ can be an isobaric hypersurface $M_p$.
But $M_p$ is indeed a smooth hypersurface of ${\mathbb R}^n$ in a neighborhood of a regular point $x\in M_p$. We say that $M_p$ is a {\it regular isobaric hypersurface} if it consists of regular points.

Recall that a vector field $u$ on a manifold $M$ with a Riemannian metric $g$ is called a {\it geodesic vector field} if
{\begin{equation}
\nabla_{\!u}\,u=0,
                  \label{2.2}
\end{equation}
}
where $\nabla$ stands for the covariant derivative with respect to the Levi-Civita connection of $(M,g)$. Integral curves of a geodesic vector field are geodesics. In the case of a GF, regular isobaric hypersurfaces $M_p\subset{\mathbb R}^n$ are considered with the Riemannian metric induced by the Euclidean metric of ${\mathbb R}^n$.

\begin{proposition}
Given a GF $(u,p)$, the restriction of $u$ to every regular isobaric hypersurface $M_p$ is a non-vanishing geodesic vector field.
\end{proposition}

\begin{proof}
It consists of one line
\begin{equation}
\nabla_{\!u}\,u=P(u\cdot\nabla u)=-P({\rm grad}\,p)=0
                          \label{2.2a}
\end{equation}
{with the following comment. $P$ is the orthogonal projection onto the tangent hyperplane of $M_p$. On the left-hand side of \eqref{2.2a}, $\nabla$ stands for the covariant derivative
with respect to the Levi-Civita connection on $M_p$.
But the second $\nabla$ stands for the Euclidean gradient, the same operator as in the Euler equations \eqref{1.1}--\eqref{1.2}.
The first equality in \eqref{2.2a} is the main relationship between intrinsic geometry of a hypersurface and geometry of the ambient space; it goes back to Gauss and is valid in a more general setting \cite[Chapter~VII, Proposition~3.1]{KN}. The second equality in \eqref{2.2a} holds by \eqref{1.1}, and the last equality holds by~\eqref{1.3}.
}
\end{proof}

Since the vector field $u$ does not vanish at regular points, no integral curve of $u$ living on a regular hypersurface $M_p$
degenerates to a point.
Thus, integral curves of $u$ constitute a geodesic foliation
of a regular part of any isobaric hypersurface.


\begin{proposition}\label{P2.2}
{Given a Gavrilov flow $(u,p)$ on an open set $U\subset{\mathbb R}^n$, l}et us restrict the vector field $u$ onto a regular isobaric hypersurface $M_p$, and let ${\rm{div}}\,u$ be the $(n-1)$-dimensional divergence of the restriction which is understood in the sense of intrinsic geometry of $M_p$. Then
\begin{equation}
 {\rm {div}}\,u = u(\log \vert {\rm grad}\,p\,\vert).
                          \label{2.3}
\end{equation}
On the right-hand side of \eqref{2.3}, the vector field $u$ is considered as a differentiation of the algebra $C^\infty(M_p)$
of smooth functions on $M_p$.
\end{proposition}

\begin{proof}
We will show that the equation \eqref{2.3} is equivalent to the incompressibility equation \eqref{1.2}. To this end we will rewrite the equation \eqref{1.2} in local curvilinear coordinates adapted to the foliation of $U$ into isobaric hypersurfaces.
Fix a regular point $x_0\in U$ and set $p_0=p(x_0)$. For $p\in{\mathbb R}$ sufficiently close to $p_0$, the isobar $M_p$ is a regular hypersurface near $x_0$. Choose local curvilinear coordinates $(z^1,\dots,z^{n-1})$ on the hypersurface $M_{p_0}$. Let
$$
{\mathbb R}^{n-1}\supset\Omega\ {\stackrel r\longrightarrow}\ {\mathbb R}^n,\quad r=r(z^1\dots,z^{n-1})
$$
be the parametrization of
$M_{p_0}$ in these coordinates. Assume that $0\in\Omega$ and $r(0)=x_0$.
In some neighborhood of $x_0$, we introduce local curvilinear coordinates $(z^1,\dots,z^n)$ in ${\mathbb R}^n$ as follows.
Define the vector field
\begin{equation}
\xi=\frac{{\rm grad}\,p}{\vert{\rm grad}\,p\,\vert^2}.
                          \label{2.4}
\end{equation}
For $(z^1,\dots,z^{n-1})\in{\mathbb R}^{n-1}$ sufficiently close to $0$, let
$$
R(z^1,\dots,z^{n-1};z^n),\quad(p_0-\varepsilon< z^n<p_0+\varepsilon)
$$
be the integral curve of the vector field $\xi$ starting at the point $r(z^1,\dots,z^{n-1})$ at the initial moment $z^n=p_0$. Thus $R(z)\in{\mathbb R}^n$ is the solution to the Cauchy problem
\begin{equation}
\frac{\partial R}{\partial z^n}(z)=\xi\big(R(z)\big),\quad R(z^1,\dots,z^{n-1};p_0)=r(z^1,\dots,z^{n-1}).
                          \label{2.5}
\end{equation}
Obviously, $R$ is a diffeomorphism between some neighborhood of the point $(0,\dots,0,p_0)$ and a neighborhood of $x_0$;
therefore $(z^1,\dots,z^n)$ constitute a local coordinate system in ${\mathbb R}^n$ near the point $x_0$.
By our construction, $R$ satisfies the identity
 $p\big(R(z)\big)=z^n$,
which means that the coordinate $z^n$
coincides with the pressure $p$. Nevertheless, we use the different notation for the coordinate since $z^n$ is considered as an independent variable while $p$ is a function on $U$.
For every $z^n$ sufficiently close to $p_0$, $(z^1,\dots,z^{n-1})$ are local coordinates on the isobaric hypersurface
$
M_{z^n}=\{x\in U\mid p(x)=z^n\}.
$
Let
$$
ds_{z^n}^2=g_{\alpha\beta}dz^\alpha dz^\beta,\quad g_{\alpha\beta}=\frac{\partial R}{\partial z^\alpha}\cdot\frac{\partial R}{\partial z^\alpha}
$$
be the first quadratic form of
$M_{z^n}$.
We use the following convention: Greek indices vary from 1 to $n-1$ and the summation from 1 to $n-1$ is assumed over a repeating Greek index; while Roman indices vary from 1 to $n$ with the corresponding summation rule.
We also write the Euclidean metric of ${\mathbb R}^n$ in coordinates $(z^1,\dots,z^n)$~as
$ds^2=h_{ij}dz^idz^j${, where} $h_{ij}=\frac{\partial R}{\partial z^i}\cdot\frac{\partial R}{\partial z^j}$.
Obviously,
$
h_{\alpha\beta}=g_{\alpha\beta},
$
and
$$
h_{\alpha n}(z)=\frac{\partial R}{\partial z^\alpha}(z)\cdot\frac{\partial R}{\partial z^n}(z)
=\frac{\partial R}{\partial z^\alpha}(z)\cdot\frac{{\rm grad}\,p}{\vert{\rm grad}\,p\,\vert^2}\big(R(z)\big)=0.
$$
The last equality holds since the vector $\frac{\partial R}{\partial z^\alpha}(z)$ is tangent to the hypersurface $M_{z^n}$ at the point $R(z)$ while the vector ${\rm grad}\,p(R(z))$ is orthogonal to $M_{z^n}$ at the same point.
Similarly, we get
$
h_{nn}(z)=\big\vert\frac{\partial R}{\partial z^n}(z)\big\vert^2 =\big\vert{\rm grad}\,p\big(R(z)\big)\big\vert^{-2}.
$
Thus,
$$
(h_{ij})
=\Big(\begin{array}{cc}g_{\alpha\beta}&0\\0&\vert{\rm grad}\,p\,\vert^{-2}\end{array}\Big),\quad
(h^{ij})=(h_{ij})^{-1}
=\Big(\begin{array}{cc}g^{\alpha\beta}&0\\0&\vert{\rm grad}\,p\,\vert^2\end{array}\Big).
$$
Let
$
\Gamma^\alpha_{\beta\gamma}=\frac{1}{2}g^{\alpha\delta}\big(\frac{\partial g_{\beta\delta}}{\partial z^\gamma}
+\frac{\partial g_{\gamma\delta}}{\partial z^\beta}-\frac{\partial g_{\beta\gamma}}{\partial z^\delta}\big)
$
be the Christoffel symbols of
$M_{z^n}$ in coordinates $(z^1,\dots,z^{n-1})$ and $G^i_{jk}$ be the Christoffel symbols of the Euclidean metric in coordinates $(z^1,\dots,z^n)$.
As follows from the above relations, the Christoffel symbols satisfy
\begin{equation}
\begin{aligned}
G^\alpha_{nn}&=-\frac{1}{2}g^{\alpha\delta}\,\frac{\partial\vert{\rm grad}\,p\,\vert^{-2}}{\partial z^\delta},\
G^n_{\beta n}=\frac{1}{2}\vert{\rm grad}\,p\,\vert^2\,\frac{\partial\vert{\rm grad}\,p\,\vert^{-2}}{\partial z^\beta},\
G^n_{nn}=\frac{1}{2}\vert{\rm grad}\,p\,\vert^2\,\frac{\partial\vert{\rm grad}\,p\,\vert^{-2}}{\partial z^n},\\
G^\alpha_{\beta\gamma}&=\Gamma^\alpha_{\beta\gamma},\quad
G^n_{\beta\gamma}=-\frac{1}{2}\vert{\rm grad}\,p\,\vert^2\,\frac{\partial g_{\beta\gamma}}{\partial z^n},\quad
G^\alpha_{\beta n}=\frac{1}{2}g^{\alpha\delta}\,\frac{\partial g_{\beta\delta}}{\partial z^n}.
\end{aligned}
                          \label{2.7}
\end{equation}
The velocity $u$ can be represented in the chosen coordinates as
 $u(z)=u^\alpha(z)\frac{\partial R(z)}{\partial z^\alpha}$,
since it is tangent to $M_{z^n}$. The velocity $u$ can be thought as a smooth vector field either on the $n$-dimensional open set $U$ or
on each isobaric hypersurface $M_{z^n}$ smoothly depending on the parameter $z^n$. We remain the notation $\nabla\cdot u$ for the $n$-dimensional divergence of $u$ while the $(n-1)$-dimensional divergence of $u$ on $M_{z^n}$ will be denoted by ${\rm {div}}\,u$.
Thus,
\begin{equation}
{\rm {div}}\,u=\nabla_{\alpha}u^{\alpha}=\frac{\partial u^\alpha}{\partial z^\alpha}+\Gamma^\alpha_{\alpha\beta}u^\beta.
                          \label{2.9}
\end{equation}
By the same formula,
$
\nabla\cdot u=\frac{\partial u^i}{\partial z^i}+G^i_{ij}u^j,
$
since $u^n=0$, this becomes
$$
\nabla\cdot u=\frac{\partial u^\alpha}{\partial z^\alpha}+G^i_{i\beta}u^\beta.
$$
In particular, the incompressibility equation \eqref{1.2} is written in the chosen coordinates as
$
\frac{\partial u^\alpha}{\partial z^\alpha}+G^i_{i\beta}u^\beta=0.
$
Using this equation, \eqref{2.9} becomes
$
{\rm {div}}\,u=-(G^i_{i\beta}-\Gamma^\alpha_{\alpha\beta})u^\beta
$.
By the formulas for Christoffel symbols,
$
G^i_{i\beta}-\Gamma^\alpha_{\alpha\beta}=-\frac{\partial(\log\vert{\rm grad}\,p\,\vert)}{\partial z^\beta}.
$
Using this expression in the previous formula {for} ${\rm {div}}\,u$, we get
$
{\rm {div}}\,u=\frac{\partial(\log\vert{\rm grad}\,p\,\vert)}{\partial z^\beta}\,u^\beta.
$
This is equivalent to~\eqref{2.3}.
\end{proof}

{{\bf Remark.} Proposition \ref{P2.2} is not completely new, compare with \cite[Theorem~3.4.12]{AM-new}.}

Let ${\rm{II}}$ be the second quadratic form of a regular isobaric hypersurface $M_p$.
{Recall that the second quadratic form depends on the choice of the unit normal vector to a hypersurface (the second quadratic form changes its sign if the unit normal vector $N$ is replaced with $-N$). We choose the unit normal vector for a regular isobaric surface to be a positive multiple of ${\rm grad}\,p$.}

\begin{proposition}
Given a GF $(u,p)$, the restriction of the vector field $u$ onto a regular isobaric hypersurface $M_p$ satisfies
\begin{equation}
{\rm{II}}(u,u)=-\vert{\rm grad}\,p\,\vert.
                          \label{2.10}
\end{equation}
\end{proposition}

\begin{proof}
{We use the same local coordinates $(z^1,\dots,z^n)$ as in the previous proof.
To prove \eqref{2.10},
we write down the equation \eqref{1.1}
as
$
(u\cdot\nabla u)^i+({\rm grad}\,p)^i=0.
$
Setting $i=\alpha$ here, we obtain nothing new; more precisely, we obtain
the same result: integral curves of $u$ are geodesics of $M_p$. Thus, we set $i=n$ in the latter formula
$$
(u\cdot\nabla u)^n+({\rm grad}\,p)^n=0.
$$
By \eqref{2.4}--\eqref{2.5},  $({\rm grad}\,p)^n=\vert{\rm grad}\,p\,\vert^2$~and
our equation becomes}
\begin{equation}
(u\cdot\nabla u)^n=-\vert{\rm grad}\,p\,\vert^2.
                          \label{2.11}
\end{equation}
By a well-known formula for covariant derivatives,
$
(u\cdot\nabla u)^n=u^i\big(\frac{\partial u^n}{\partial z^i}+G^n_{ij}u^j\big).
$
Since $u^n=0$, this becomes
$
(u{\cdot}\nabla u)^n\!=G^n_{\alpha\beta}u^\alpha u^\beta.
$
Using
$G^n_{\alpha\beta}\!=\!-\frac{1}{2}\,\vert{\rm grad}\,p\,\vert^2\,\frac{\partial g_{\alpha\beta}}{\partial z^n}$ from \eqref{2.7}, we~get
\begin{equation}
(u\cdot\nabla u)^n=-\frac{1}{2}\,\vert{\rm grad}\,p\,\vert^2\,\frac{\partial g_{\alpha\beta}}{\partial z^n}u^\alpha u^\beta.
                          \label{2.12}
\end{equation}
Differentiating the equality
$
g_{\alpha\beta}=\frac{\partial R}{\partial z^\alpha}\cdot\frac{\partial R}{\partial z^\beta}
$
with respect to $z^n$, we obtain
{$$
\frac{\partial g_{\alpha\beta}}{\partial z^n}
=\frac{\partial^2 R}{\partial z^\alpha\partial z^n}\cdot\frac{\partial R}{\partial z^\beta}
+\frac{\partial^2 R}{\partial z^\beta\partial z^n}\cdot\frac{\partial R}{\partial z^\alpha}.
$$}
This can be written as
$$
\frac{\partial g_{\alpha\beta}}{\partial z^n}
=\frac{\partial}{\partial z^\alpha}\Big(\frac{\partial R}{\partial z^n}\cdot\frac{\partial R}{\partial z^\beta}\Big)
+\frac{\partial}{\partial z^\beta}\Big(\frac{\partial R}{\partial z^n}\cdot\frac{\partial R}{\partial z^\alpha}\Big)
-2\frac{\partial^2 R}{\partial z^\alpha\partial z^\beta}\cdot\frac{\partial R}{\partial z^n}.
$$
Both expressions in parentheses are equal to zero, and we obtain
\begin{equation}
\frac{\partial g_{\alpha\beta}}{\partial z^n}
=-2\frac{\partial^2 R}{\partial z^\alpha\partial z^\beta}\cdot\frac{\partial R}{\partial z^n}.
                          \label{2.13}
\end{equation}
Let
$
N=\frac{{\rm grad}\,p}{\vert{\rm grad}\,p\,\vert}=\vert{\rm grad}\,p\,\vert\,\frac{\partial R}{\partial z^n}
$
be the unit normal vector of the hypersurface $M_{z^n}$.
{By~classic formulas of differential geometry (so called {\it derived formulas} \cite{P}),
$$
\frac{\partial^2R}{\partial z^\alpha\partial z^\beta}
=\Gamma^\gamma_{\alpha\beta}\frac{\partial R}{\partial z^\gamma}+b_{\alpha\beta}N,
$$
where $b_{\alpha\beta}$ are
coefficients of the second quadratic form}
for $M_{z^n}$ in coordinates $(z^1,\dots,z^{n-1})$. Taking the scalar product of this equality with $N$ and using the orthogonality of $N$ to $\frac{\partial R}{\partial z^\gamma}$, we obtain
$\frac{\partial^2R}{\partial z^\alpha\partial z^\beta}\cdot N=b_{\alpha\beta}$.
Since $N=\vert{\rm grad}\,p\,\vert\,\frac{\partial R}{\partial z^n}$, we get
$
\frac{\partial^2R}{\partial z^\alpha\partial z^\beta}\cdot \frac{\partial R}{\partial z^n}=\vert{\rm grad}\,p\,\vert^{-1}b_{\alpha\beta}
$.
Using this equality, formula \eqref{2.13} becomes
$
\frac{\partial g_{\alpha\beta}}{\partial z^n}=-2\vert{\rm grad}\,p\,\vert^{-1}b_{\alpha\beta}
$.
Substituting this value into \eqref{2.12}, we get
$$
(u\cdot\nabla u)^n=\vert{\rm grad}\,p\,\vert\,b_{\alpha\beta}u^\alpha u^\beta=\vert{\rm grad}\,p\,\vert\,{\rm{II}}(u,u).
$$
Inserting this expression into \eqref{2.11}, we arrive to \eqref{2.10}.
\end{proof}

{Although some special coordinates have been used in the proof, the resulting equations \eqref{2.3} and \eqref{2.10} are of an invariant nature, i.e., independent of a coordinates choice.}

All our arguments in this section are invertible, i.e., the following statement is valid.

\begin{proposition}
{The Euler -- Gavrilov system \eqref{1.1}--\eqref{1.3} is equivalent to the system \eqref{2.2}, \eqref{2.3}, \eqref{2.10}. More precisely, let a smooth vector field $u$ and smooth real function $p$ be defined on an open set $U\subset{\mathbb R}^n$. Choose a point $x_0\in U$ such that ${\rm grad}\,p(x_0)\ne 0$. Assume that integral curves of $u$ are tangent to level hypersurfaces $M_p$ and equations \eqref{2.2}, \eqref{2.3},  \eqref{2.10} hold in some neighborhood of $x_0$. Then $(u,p)$ is a Gavrilov flow in some neighborhood of $x_0$}.
\end{proposition}


{Now, we discuss an easy example of GF which exists in any even dimension.}
%
Let $(x_1,\dots,x_{2n})$ be Cartesian coordinates in ${\mathbb R}^{2n}$. Set
\begin{equation}
u_{2j-1}(x)=-x_{2j},\ u_{2j}(x)=x_{2j-1}\ (j=1,\dots, n),\quad
p(x)=\frac{1}{2}\,\vert x\vert^2.
                          \label{3.1}
\end{equation}
It is easy to check that the equations \eqref{1.1}--\eqref{1.3} hold in ${\mathbb R}^{2n}$. Observe that $\vert u\vert^2=2p$.
{Every $x\ne0$ is a regular point}. Isobaric hypersurfaces are spheres $M_p=\{x\in{\mathbb R}^{2n}: \vert x\vert^2=2p\}$.
Integral curves of
$u$ (particle trajectories) are circles centered at the origin. Every sphere $M_p$ is foliated by particle trajectories. This foliation coincides with the {well-known}
Hopf fiber bundle ${\mathbb S}^{2n-1}\rightarrow{\mathbb C}P^{n-1}$ of an odd-dimensional sphere over the complex projective~space.

{A} GF on ${\mathbb R}^{2n+1}$ can be obtained as a direct product of the flow \eqref{3.1} with a constant velocity flow. Namely,
\begin{equation}
u_{2j-1}(x)=-x_{2j},\ u_{2j}(x)=x_{2j-1}\ (j=1,\dots, n),\quad u_{2n+1}(x)=a=\mbox{const}
                          \label{3.2}
\end{equation}
and $p(x)=\frac{1}{2}(x_1^2+\dots+x_n^2)$. Isobaric hypersurfaces are cylinders ${\mathbb S}^{2n-1}\times{\mathbb R}$, and particle trajectories are either circles ({if} $a=0$) or helices
({if} $a\ne0$).

In {order} to apply the Gavrilov localization to the flow \eqref{3.1}, choose a compactly supported smooth function
$\alpha:[0,\infty)\rightarrow{\mathbb R}$ such that $\alpha(r)=0$ for $r\le\varepsilon$ with some $\varepsilon>0$ and define the function $\beta:[0,\infty)\rightarrow{\mathbb R}$  by $\beta(r)=-\int_r^\infty s\,\alpha^2(s)\,ds$. Then
$$
\tilde u(x)=\alpha(\vert x\vert)u(x),\quad \tilde p(x)=\beta(\vert x\vert)
$$
is a smooth compactly supported GF on {the whole of} ${\mathbb R}^{2n}$ satisfying $\vert\tilde u\vert^2=\psi(\tilde p)$ with a function $\psi$ uniquely determined by $\alpha$. In particular, if $\alpha$ is supported in $(r_0-\delta,r_0+\delta)$ for some $r_0>\delta>0$, then the velocity
$\tilde u$ is supported in the spherical layer $\{x\in{\mathbb R}^{2n}: r_0-\delta<\vert x\vert<r_0+\delta\}$, and the pressure $\tilde p$ is supported in the ball
$\{x\in{\mathbb R}^{2n}: \vert x\vert<r_0+\delta\}$ with $\tilde p=\mbox{const}$ in the smaller ball $\{x\in{\mathbb R}^{2n}: \vert x\vert\le r_0-\delta\}$.
Then we can take a linear combination of several such localized flows with disjoints supports.
In particular, a periodic GF can be constructed in this way.

\section{Plane sections of a Gavrilov flow}

Let a {sufficiently} smooth solution $(u,p)$ of the Euler equations \eqref{1.1}--\eqref{1.2} be defined on {the whole of} ${\mathbb R}^n$. Assume the solution to decay sufficiently fast at infinity together with first order derivatives (e.g., it can be a smooth compactly supported solution).
Then {\cite{NV}} the equality
\begin{equation}
\int\nolimits_P\big(\xi\cdot u(x)\big)\big(\nu\cdot u(x)\big)\,dx=0
                                   \label{4.1}
\end{equation}
holds for every affine hyperplane $P\subset{\mathbb R}^n$ and every vector $\xi\in{\mathbb R}^n$ parallel to $P$, where $\nu$ is the normal vector to the hyperplane $P$ and $dx$ stands for the $(n-1)$-dimensional Lebesgue measure on $P$. Actually there are $n-1$ independent equations in \eqref{4.1} since the vector $\xi$ can take $n-1$ linearly independent values from the space $\nu^\bot=\{\xi\in{\mathbb R}^n\mid \nu\cdot\xi=0\}$.

For a GF, the equation \eqref{4.1}, combined with the Gavrilov localization, yields an interesting statement.
The following theorem can be easily generalized to an arbitrary dimension.

\begin{theorem}
Let $(u,p)$ be a smooth GF defined on {the whole of} ${\mathbb R}^3$ and sufficiently fast decaying at infinity together with first order derivatives.
Let $M_{p_0}$ be a regular isobaric surface and let an affine plane $P_0$ transversally intersect $M_{p_0}$. Then, for any $p$ sufficiently close to $p_0$ and for any affine plane $P$ sufficiently close to $P_0$, 
\begin{equation}
\int\limits_{M_p\cap P}\frac{1}{\vert{\rm grad}\,q(x)\vert}\big(\xi\cdot u(x)\big)\big(\nu\cdot u(x)\big)\,ds=0,
                                   \label{4.2}
\end{equation}
where $q\in C^\infty(P)$ is the restriction of the function $p$ to the plane $P$, $\nu$ is the unit normal vector to $P$, and $\xi$ is an arbitrary vector parallel to $P$. The integration in \eqref{4.2} is performed with respect to the arc length $ds$ of the curve $M_p\cap P$.
\end{theorem}

\begin{proof}
Since $M_{p_0}$ and $P_0$ intersect transversally, the same is true for $M_p$ and $P$ for any $p$ close to $p_0$ and for any plane $P$ close to $P_0$. We fix such $p$ and $P$, set $q=p\,\vert_P$ and $\gamma=M_p\cap P$. Observe that $\gamma$ is a regular curve on the plane $P$ and the gradient ${\rm grad}\,q$ does not vanish in some neighborhood of $\gamma$. Therefore the integral {on the left-hand side of}
\eqref{4.2} is well defined.

We parameterize the curve $\gamma$ by the arc length, $\gamma=\gamma(s)$. Then we choose local coordinates $(s,\tau)$ in a neighborhood $U\subset P$ of $\gamma$ in the same way as in the proof of {Proposition~\ref{P2.2}}. Namely, the coordinates are chosen so that $x(s,0)=\gamma(s)$ and $q(x(s,\tau))=p+\tau$. For every $s_0$, the coordinate line $\delta(\tau)=x(s_0,\tau)$ starts at $\gamma(s_0)$ orthogonally to $\gamma$ with the initial speed $\vert\dot\delta(0)\vert=\frac{1}{\vert\nabla q(\gamma(s_0))\vert}$.
Therefore the area form $dx$ of the plane $P$ is written in the chosen coordinates as
$dx=\big(\frac{1}{\vert{\rm grad}\,q(\gamma(s))\vert}+o(\tau)\big)ds\,d\tau$.

Fix a smooth function $\mu:{\mathbb R}\rightarrow{\mathbb R}$ such that $\mu(r)=0$ for $\vert r\vert\ge1$, $\mu(r)>0$ for $\vert r\vert<1$, and $\int_{-1}^1\mu(r)\,dr=1$. For $\varepsilon>0$, set $\alpha_\varepsilon(r)=\sqrt{\mu((r-c)/\varepsilon)}$. Using the latter function, we define the localized GF $(\tilde u,\tilde p)$ by
$\tilde u=\alpha_\varepsilon(p)\,u,\ {\rm grad}\,\tilde p=\alpha_\varepsilon^2(p)\,{\rm grad}\,p$.
Applying \eqref{4.1} to $(\tilde u,\tilde p)$, we obtain
$\int_P\alpha_\varepsilon(q(x))\big(\xi\cdot u(x)\big)\big(\nu\cdot u(x)\big)\,dx=0$.
In the chosen coordinates, this is written as
\begin{equation}
\int\limits_\gamma\int\limits_{-\varepsilon}^\varepsilon \mu(\tau/\varepsilon)\Big(\frac{1}{\vert{\rm grad}\,q(\gamma(s))\vert}+o(\tau)\Big)
\big(\xi\cdot u(x(s,\tau))\big)\big(m\cdot u(x(s,\tau))\big)\,d\tau\,ds=0.
                                   \label{4.3}
\end{equation}
The integrand can be represented as
\[
\begin{aligned}
& \mu(\tau/\varepsilon)\Big(\frac{1}{\vert{\rm grad}\,q(\gamma(s))\vert}+o(\tau)\Big)
\big(\xi\cdot u(x(s,\tau))\big)\big(\nu\cdot u(x(s,\tau))\big)\\
&\qquad =\mu(\tau/\varepsilon)\frac{1}{\vert{\rm grad}\,q(\gamma(s))\vert}
\big(\xi\cdot u(\gamma(s))\big)\big(\nu\cdot u(\gamma(s))\big)+o(1).
\end{aligned}
\]
The variables $s$ and $\tau$ are separated up to $o(\tau)$ on the right-hand side of the latter formula.
Using this representation and $\int_{-\varepsilon}^\varepsilon\mu(\tau/\varepsilon)\,d\tau=\varepsilon$, we obtain from \eqref{4.3}
$\varepsilon\int\limits_\gamma\frac{1}{\vert{\rm grad}\,q(\gamma(s))\vert}
\big(\xi\cdot u(\gamma(s))\big)\big(\nu\cdot u(\gamma(s))\big)\,ds+o(\varepsilon)=0.
$
In the limit as $\varepsilon\rightarrow0$, this gives~\eqref{4.2}.
\end{proof}

\section{Differential equations for {a Gavrilov flow}}
\label{sec:eqs}

{We consider the 3D case in this section. Let a GF $(u,p)$ be defined on an open set of ${\mathbb R}^3$. Let $(x,y,z)$ be Cartesian coordinates. We assume that, for $p\in(-p_0,p_0)$, the regular isobar surface $M_p$ coincides with the graph of a smooth function
\begin{equation}
z=f(p;x,y)\quad \big((x,y)\in U\big)
                                    \label{5.1}
\end{equation}
for some open domain $U\subset{\mathbb R}^2$.}
The first and the second quadratic forms of {$M_p$} are
\begin{eqnarray}
\nonumber
&& I =(1+f'_x{}^2)dx^2+2f'_xf'_y\,dxdy+(1+f'_y{}^2)dy^2,\\
&& II=\pm\frac{1}{\sqrt{1+f'_x{}^2+f'_y{}^2}}\big(f''_{xx}\,dx^2+2f''_{xy}\,dxdy+f''_{yy}\,dy^2\big),
                                    \label{5.3}
\end{eqnarray}
where the sign depends on the choice of the unit vector normal to $M_p$.
Christoffel symbols of this metric are
\begin{equation}
\begin{aligned}
\Gamma^x_{xx}&=\frac{f'_xf''_{xx}}{1+f'_x{}^2+f'_y{}^2},\quad
\Gamma^x_{xy}=\frac{f'_xf''_{xy}}{1+f'_x{}^2+f'_y{}^2},\quad
\Gamma^x_{yy}=\frac{f'_xf''_{yy}}{1+f'_x{}^2+f'_y{}^2},\\
\Gamma^y_{xx}&=\frac{f'_yf''_{xx}}{1+f'_x{}^2+f'_y{}^2},\quad
\Gamma^y_{xy}=\frac{f'_yf''_{xy}}{1+f'_x{}^2+f'_y{}^2},\quad
\Gamma^y_{yy}=\frac{f'_yf''_{yy}}{1+f'_x{}^2+f'_y{}^2}.
\end{aligned}
                                    \label{5.4}
\end{equation}
Let $(u^x,u^y)$ be geometric coordinates of the vector field $u$, i.e.,
$u=u^x(p;x,y)\frac{\partial}{\partial x}+u^y(p;x,y)\frac{\partial}{\partial y}$,
where $\frac{\partial}{\partial x}$ and $\frac{\partial}{\partial y}$ are considered as coordinate vector fields tangent to the surface $M_p$.
Using \eqref{5.4}, we compute covariant derivatives
\begin{equation}
\begin{aligned}
\nabla_{\!x}u^x&=\frac{\partial u^x}{\partial x}+\frac{f'_x (f''_{xx}u^x+f''_{xy}u^y)}{1+f'_x{}^2\!+\!f'_y{}^2},\quad
\nabla_{\!y}u^x=\frac{\partial u^x}{\partial y}+\frac{f'_x (f''_{xy}u^x+f''_{yy}u^y)}{1+f'_x{}^2\!+\!f'_y{}^2},\\
\nabla_{\!x}u^y&=\frac{\partial u^y}{\partial y}+\frac{f'_y (f''_{xx}u^x+f''_{xy}u^y)}{1+f'_x{}^2\!+\!f'_y{}^2},\quad
\nabla_{\!y}u^y=\frac{\partial u^y}{\partial y}+\frac{f'_y (f''_{xy}u^x+f''_{yy}u^y)}{1+f'_x{}^2\!+\!f'_y{}^2}.
\end{aligned}
                                    \label{5.7}
\end{equation}
In particular,
\begin{equation}
\mbox{div}\,u=\nabla_{\!x}u^x{+}\nabla_{\!y}u^y=
\frac{\partial u^x}{\partial x}{+}\frac{\partial u^y}{\partial y}
+\frac{(f'_xf''_{xx}{+}f'_yf''_{xy})u^x{+}(f'_xf''_{xy}{+}f'_yf''_{yy})u^y}{1\!+\!f'_x{}^2\!+\!f'_y{}^2}.
                                    \label{5.8}
\end{equation}
Substituting expressions \eqref{5.7} into the formula
$
\nabla_{\!u}u=(u^x\nabla_{\!x}u^x+u^y\nabla_{\!y}u^x)\frac{\partial}{\partial x} +(u^x\nabla_{\!x}u^y+u^y\nabla_{\!y}u^y)\frac{\partial}{\partial y},
$
we get
$$
\begin{aligned}
\nabla_{\!u}u&=\Big(u^x\frac{\partial u^x}{\partial x}+u^y\frac{\partial u^x}{\partial y}
{+}\frac{f'_x}{1\!+\!f'_x{}^2\!+\!f'_y{}^2}\big(f''_{xx}(u^x)^2{+}2f''_{xy}u^xu^y{+}f''_{yy}(u^y)^2\big)\Big)
\frac{\partial}{\partial x}\\
&+\Big(u^x\frac{\partial u^y}{\partial x}+u^y\,\frac{\partial u^y}{\partial y}
{+}\frac{f'_y}{1\!+\!f'_x{}^2\!+\!f'_y{}^2}\big(f''_{xx}(u^x)^2{+}2f''_{xy}u^xu^y{+}f''_{yy}(u^y)^2\big)\Big)
\frac{\partial}{\partial y}\,.
\end{aligned}
$$
Being a geodesic vector field, $u$ {solves the equation} $\nabla_{\!u}u=0$. We {thus} arrive to the system
\begin{equation}
\begin{aligned}
u^x\,\frac{\partial u^x}{\partial x}+u^y\,\frac{\partial u^x}{\partial y}
+\frac{f'_x}{1\!+\!f'_x{}^2\!+\!f'_y{}^2}\big(f''_{xx}(u^x)^2+2f''_{xy}u^xu^y+f''_{yy}(u^y)^2\big)&=0,\\
u^x\,\frac{\partial u^y}{\partial x}+u^y\,\frac{\partial u^y}{\partial y}
+\frac{f'_y}{1\!+\!f'_x{}^2\!+\!f'_y{}^2}\big(f''_{xx}(u^x)^2+2f''_{xy}u^xu^y+f''_{yy}(u^y)^2\big)&=0.
\end{aligned}
                                    \label{5.9}
\end{equation}
Let us express $\vert{\rm grad}\,p\,\vert$ in terms of the function $f$. Let $p(x,y,z)$ be the pressure in Cartesian coordinates.
We have the identity $p(x,y,f(p;x,y))=p$. Differentiate the identity to get
$$
\begin{aligned}
p'_x(x,y,f(p;x,y))+p'_z(x,y,f(p;x,y))f'_x(p;x,y)&=0,\\
p'_y(x,y,f(p;x,y))+p'_z(x,y,f(p;x,y))f'_y(p;x,y)&=0,\\
p'_z(x,y,f(p;x,y))f'_p(p;x,y)&=1.
\end{aligned}
$$
From this we get
$p'_x=-\frac{f'_x}{f'_p},\ p'_y=-\frac{f'_y}{f'_p},\ p'_z=\frac{1}{f'_p}$.
The derivative $f'_p$ does not vanish on a regular isobaric surface \eqref{5.1}. Thus,
\begin{equation}
\vert{\rm grad}\,p\,\vert^2=p'_x{}^2+p'_y{}^2+p'_z{}^2=\frac{1+f'_x{}^2+f'_y{}^2}{f'_p{}^2}.
                                    \label{5.11}
\end{equation}
Substituting expressions \eqref{5.8} and \eqref{5.11} into \eqref{2.3}, we arrive to the equation
$$
\begin{aligned}
&\frac{\partial u^x}{\partial x}+\frac{\partial u^y}{\partial y}
+\frac{(f'_xf''_{xx}+f'_yf''_{xy})u^x+(f'_xf''_{xy}+f'_yf''_{yy})u^y}{1\!+\!f'_x{}^2\!+\!f'_y{}^2}\\
&=\frac{1}{2}\cdot\frac{f'_p{}^2}{1+f'_x{}^2+f'_y{}^2}
\Big[u^x\frac{\partial}{\partial x}\Big(\frac{1+f'_x{}^2+f'_y{}^2}{f'_p{}^2}\Big)
+u^y\frac{\partial}{\partial y}\Big(\frac{1+f'_x{}^2+f'_y{}^2}{f'_p{}^2}\Big)\Big].
\end{aligned}
$$
After the obvious simplification, it becomes
\[
\frac{\partial u^x}{\partial x}+\frac{\partial u^y}{\partial y}+\frac{1}{f'_p}(u^xf''_{px}+u^yf''_{py})=0.
\]
Substituting expressions \eqref{5.3} and \eqref{5.11} into \eqref{2.10}, we arrive to the equation
$f''_{xx}(u^x)^2+2f''_{xy}u^xu^y+f''_{yy}(u^y)^2=\frac{1+f'_x{}^2+f'_y{}^2}{\vert f'_p\vert}$.
Using this, equations \eqref{5.9} can be written as
$$
u^x\,\frac{\partial u^x}{\partial x}+u^y\,\frac{\partial u^x}{\partial y}+\frac{f'_x}{\vert f'_p\vert}=0,\quad
u^x\,\frac{\partial u^y}{\partial x}+u^y\,\frac{\partial u^y}{\partial y}
+\frac{f'_y}{\vert f'_p\vert}=0.
$$
We have proved the following

\begin{proposition} \label{P5.1}
Let a GF $(u,p)$ be defined on an open set of ${\mathbb R}^3$.
Assume that for $p\in(-p_0,p_0)$ the isobaric surface $M_p$ is regular and coincides with the graph of a smooth function
$z=f(p;x,y)\ \big((x,y)\in U\subset{\mathbb R}^2\big)$.
Write the restriction of the vector field $u$ to the surface $M_p$ in the form
$u=u^x(p;x,y)\frac{\partial}{\partial x}+u^y(p;x,y)\frac{\partial}{\partial y}$.
Then the derivative $f'_p$ does not vanish and the functions $f(p;x,y),u^x(p;x,y)$ and $u^y(p;x,y)$ satisfy the equations
\begin{eqnarray}
\frac{\partial (f'_pu^x)}{\partial x}+\frac{\partial (f'_pu^y)}{\partial y}=0,
                                    \label{5.12}\\
u^x\,\frac{\partial u^x}{\partial x}+u^y\,\frac{\partial u^x}{\partial y}+\frac{f'_x}{\vert f'_p\vert}=0,
                                    \label{5.13} \\
u^x\,\frac{\partial u^y}{\partial x}+u^y\,\frac{\partial u^y}{\partial y}+\frac{f'_y}{\vert f'_p\vert}=0,
                                    \label{5.14}\\
f''_{xx}(u^x)^2+2f''_{xy}u^xu^y+f''_{yy}(u^y)^2=\frac{1+f'_x{}^2+f'_y{}^2}{\vert f'_p\vert}.
                                    \label{5.15}
\end{eqnarray}
\end{proposition}

It is easy to check that the system \eqref{5.12}--\eqref{5.15} has the following solution:
\begin{eqnarray}
f(x,y,p)=f(x,p)=-\sqrt{r^2(p)-x^2}\quad\big(-r(p)<x<r(p)\big),
                                    \label{5.16}\\
u^x(x,y,p)=u^x(x,p)=\frac{\sqrt{r^2(p)-x^2}}{\sqrt{r(p)\vert r'(p)\vert}},\quad u^y(x,y,p)=u^y(p)=b(p),
                                    \label{5.17}
\end{eqnarray}
where $r(p)$ is a smooth positive function with non-vanishing derivative, and $b(p)$ is an arbitrary smooth function.
For every $p$, the graph $M_p$ of the function $(x,y)\mapsto f(x,p)$ is the half of the cylinder ${\tilde M}_p=\{(x,y,z)\mid x^2+z^2=r^2(p)\}$. Observe that ${\tilde M}_p$ is a surface of revolution around the $y$-axis. Thus, \eqref{5.16}--\eqref{5.17} is an axisymmetric GF.
{We will study axisymmetric Gavrilov flows in the next section. The solution \eqref{5.16}--\eqref{5.17} can be slightly modified by a change of Cartesian coordinates in ${\mathbb R}^3$. We do not know any solution to the system \eqref{5.13}--\eqref{5.16} different of the (modified) solution \eqref{5.16}--\eqref{5.17}.}

\section{Axisymmetric {Gavrilov flows}}
\label{sec:axi}

Let $(r,z,\theta)$ be cylindrical coordinates in ${\mathbb R}^3$ related to Cartesian coordinates $(x_1,x_2,x_3)$ by
$
x_1=r\cos\theta,\ x_2=r\sin\theta,\ x_3=z.
$
We study a GF $(u,p)$ invariant under rotations around the $z$-axis. The flow is defined in an open set
$\tilde U\subset\{(r,z,\theta): r>0\}$ invariant under rotations around the $z$-axis.
Such a rotationally invariant set $\tilde U$ is uniquely determined by the two-dimensional set $U=\tilde U\cap\{\theta=0\}\subset\{(r,z): r>0\}$. For brevity we say that an axisymmetric GF is defined in $U$.

A regular isobaric surface $M_p$ is a surface of revolution determined by its {\it generatrix} $\Gamma_p=M_p\cap\,U$.
We parameterize the curve $\Gamma_p$ by the arc length
$r=R(p,t)>0,\ z=Z(p,t)$,
\begin{equation}
R'_t{}^2+Z'_t{}^2=1.
                                          \label{6.1}
\end{equation}
The variables $(t,\theta)$ serve as coordinates on the isobaric surface $M_p$.
Since the vector field $u$ is tangent to $M_p$, it is uniquely represented as
$u=u^t(p,t)\frac{\partial}{\partial t}+u^\theta(p,t)\frac{\partial}{\partial\theta}$,
where
$\big(u^t(p,t),u^\theta(p,t)\big)$ are {\it geometric coordinates} of $u$.
{\it Physical coordinates} of $u$ are {defined} by
\begin{equation}
u=u_r(r,z)e_r+u_z(r,z)e_z+u_\theta(r,z)e_\theta,
                                          \label{6.2}
\end{equation}
where $e_r,e_z,e_\theta$ are unit coordinate vectors. {Physical and geometric coordinates are related by}
\begin{equation}
u_r=R'_tu^t,\quad u_z=Z'_tu^t,\quad u_\theta=Ru^\theta.
                                          \label{6.3}
\end{equation}
We are going to write down differential equations for a GF in terms of the functions $(R,Z,u^t, u^\theta)$.
{First of all}, the first quadratic form of $M_p$ in coordinates $(t,\theta)$ is
\begin{equation}
I=dt^2+R^2d\theta^2.
                                          \label{6.4}
\end{equation}
Christoffel symbols of the metric are
$\Gamma^t_{\theta\theta}=-RR'_t,\
\Gamma^\theta_{t\theta}=\frac{R'_t}{R},\
\Gamma^t_{tt}=\Gamma^t_{t\theta}=\Gamma^\theta_{tt}=\Gamma^\theta_{\theta\theta}=0$.
Using these formulas, we calculate
$$
\nabla_{\!t}u^t=\frac{\partial u^t}{\partial t},\quad
\nabla_{\!\theta}u^t=-RR'_t\,u^\theta,\quad
\nabla_{\!t}u^\theta=\frac{\partial u^\theta}{\partial t}+\frac{R'_t}{R}\,u^\theta,\quad
\nabla_{\!\theta}u^\theta=\frac{R'_t}{R}\,u^t.
$$
In particular,
\begin{equation}
\mbox{div}\,u=\nabla_{\!t}u^t+\nabla_{\!\theta}u^\theta
=\frac{\partial u^t}{\partial t}+\frac{R'_t}{R}\,u^t.
                                          \label{6.6}
\end{equation}
The restriction of $u$ to $M_p$ is a geodesic vector field, i.e., $\nabla_{\!u}u=0$.
This gives the system
\begin{equation}\label{6.7}
 u^t\,\frac{\partial u^t}{\partial t}-RR'_t(u^\theta)^2=0,\quad
 u^t\,\frac{\partial u^\theta}{\partial t}+\frac{2R'_t}{R}\,u^tu^\theta=0.
\end{equation}
Let us express $\vert{\rm grad}\,p\,\vert$ in terms of the functions $R(p,t)$ and $Z(p,t)$. Let $p=p(r,z)$ be the pressure in cylindric coordinates (it is independent of $\theta$). We have the identity
$
 p\big(R(p,t),Z(p,t)\big)=p.
$
Differentiating the identity with respect to $p$ and $t$, we arrive to the linear algebraic system with unknowns $p'_r\big(R(p,t),Z(p,t)\big)$ and $p'_z\big(R(p,t),Z(p,t)\big)$
$$
\begin{aligned}
R'_p(p,t)p'_r\big(R(p,t),Z(p,t)\big)+Z'_p(p,t)p'_z\big(R(p,t),Z(p,t)\big)&=1,\\
R'_t(p,t)p'_r\big(R(p,t),Z(p,t)\big)+Z'_t(p,t)p'_z\big(R(p,t),Z(p,t)\big)&=0.
\end{aligned}
$$
Solving the system, we have
$$
p'_r\big(R(p,t),Z(p,t)\big)=J^{-1}(p,t)Z'_t(p,t),\ p'_z\big(R(p,t),Z(p,t)\big)=-J^{-1}(p,t)R'_t(p,t),
$$
where
\begin{equation}
J=\left\vert\begin{array}{cc}R'_p&R'_t\\Z'_p&Z'_t\end{array}\right\vert
                                          \label{6.9}
\end{equation}
is the Jacobian of the transformation
$r=R(p,t),\quad z=Z(p,t)$.
The Jacobian does not vanish. This implies with the help of \eqref{6.1}
\begin{equation}
\vert{\rm grad}\,p\,\vert=\vert J\vert^{-1}.
                                          \label{6.11}
\end{equation}
Substituting expressions \eqref{6.6} and \eqref{6.11} into \eqref{2.3}, we arrive to the equation
\begin{equation}
\frac{\partial u^t}{\partial t}+\Big(\frac{R'_t}{R}+\frac{J'_t}{J}\Big)u^t=0.
                                          \label{6.12}
\end{equation}
The second quadratic form of the surface $M_p$ is expressed in coordinates $(t,\theta)$ by
$$
II=-((R'_tZ''_{tt}-Z'_tR''_{tt})\,dt^2+RZ'_t\,d\theta^2).
$$
The sign on the right-hand side is chosen taking our agreement into account: the unit normal vector to the surface $M_p$ must coincide with
$\frac{{\rm grad}\,p}{\vert{\rm grad}\,p\,\vert}$.
Recall that the curvature $\kappa=\kappa(p,t)$ of the plane curve $r=R(p,t),z=Z(p,t)$ is expressed, under the condition \eqref{6.1}, by
\begin{equation}
\kappa=R'_tZ''_{tt}-Z'_tR''_{tt}.
                                          \label{6.13}
\end{equation}
Using the latter equality, the previous formula takes the form
\begin{equation}
II=-(\kappa\,dt^2+RZ'_t\,d\theta^2).
                                          \label{6.14}
\end{equation}
By \eqref{6.11} and \eqref{6.14}, the equation \eqref{2.10} takes the form
\begin{equation}
\kappa(u^t)^2+RZ'_t(u^\theta)^2=\vert J\vert^{-1}.
                                          \label{6.15}
\end{equation}
We {have thus} obtained the system of five equations \eqref{6.1}, (\ref{6.7}a,b), \eqref{6.12}, \eqref{6.15} in four unknown functions $(R,Z,u^t,u^\theta)$. The functions $J$ and $\kappa$ participating in the system are expressed through $(R,Z)$ by \eqref{6.9} and \eqref{6.13} respectively.
We proceed to the analysis of the system.

All isobaric surfaces $M_p$ under consideration are assumed to be regular and connected. The equation \eqref{6.12} implies the following alternative for every $p_0$: either $u^t(p_0,t)\neq0$ for all $t$ or $u^t(p_0,t)\equiv0$. The second case of the alternative is realized in the example \eqref{3.2} with $a=0$. The converse statement is true at least partially: If $u^t(p_0,t)\equiv0$, then $M_{p_0}$ coincides with the cylinder $\{r=\mbox{const}>0\}$ and particle trajectories living on $M_{p_0}$ are horizontal circles. Indeed, if $u^t(p_0,t)\equiv0$ then, as is seen from \eqref{6.1} and (\ref{6.7}b), $R'_t(p_0,t)\equiv0$ and $Z'_t(p_0,t)\equiv\pm1$. Nevertheless, it is possible that $u^t(p_0,t)\equiv0$ but $u^t(p,t)\neq0$ for $p$ close to $p_0$; the corresponding example can be constructed by a slight modification of~\eqref{3.2}.

To avoid degenerate cases of the previous paragraph, we additionally assume $u^t(p,t)\neq0$ for all $(p,t)$.
Recall that we study flows up to the equivalence \eqref{1.4}.
{In particular, GFs $(u,p)$ and $(-u,p)$ are equivalent.}
Therefore the latter assumption can be written without lost of generality in the form
\begin{equation}
u^t(p,t)>0\quad\mbox{for all}\quad(p,t).
                                          \label{6.16}
\end{equation}
The equation (\ref{6.7}a) simplifies under the assumption \eqref{6.16} to the following {one}:
\begin{equation}
\frac{\partial u^\theta}{\partial t}+\frac{2R'_t}{R}\,u^\theta=0.
                                          \label{6.17}
\end{equation}
If $u^\theta(p,t_0)=0$ for some $t_0$, then \eqref{6.17} implies that $u^\theta(p,t)=0$ for all $t$.
On the other hand, assuming that $u^\theta(p,t)\neq0$ for all $t$ and for a fixed $p$, we can rewrite \eqref{6.17} in the form
$\frac{\partial(\log u^\theta)}{\partial t}+\frac{\partial(\log R^2)}{\partial t}=0$.
From this we get
\begin{equation}
u^\theta(p,t)=\frac{b(p)}{R^2(p,t)}
                                          \label{6.18}
\end{equation}
with some function $b(p)$. This equality is also true for such $p_0$ that $u^\theta(p_0,t)=0$ for all $t$, just by setting $b(p_0)=0$.
The equality \eqref{6.18} implies smoothness of the function $b$.

Substituting the expression \eqref{6.18} into (\ref{6.7}a), we get
$\frac{\partial (u^t)^2}{\partial t}=2\,\frac{R'_t}{R^3}b^2(p)$,
that can be written in the form
$\frac{\partial (u^t)^2}{\partial t}=-\frac{\partial}{\partial t}\big(\frac{b^2(p)}{R^2}\big)$.
From this we obtain
\begin{equation}
u^t(p,t)=\frac{\sqrt{d(p)R^2(p,t)-b^2(p)}}{R(p,t)},
                                          \label{6.19}
\end{equation}
where $d(p)$ is a smooth function satisfying
\begin{equation}
d(p)R^2(p,t)-b^2(p)>0.
                                          \label{6.20}
\end{equation}
By  \eqref{6.4},
$\vert u\vert^2=(u^t)^2+R^2(u^\theta)^2$.
Substituting values \eqref{6.18}--\eqref{6.19}, we get
\begin{equation}
\vert u(p,t)\vert^2=d(p).
                                          \label{6.21}
\end{equation}

{We have thus discovered the important phenomenon:}
for an axisymmetric GF, all particles living on an isobaric surface $M_p$ move with the same speed. In other words, constants $c$ and $C$ in the Bernoulli law \eqref{2.1} can be expressed through each other.
{The phenomenon is actually expected since it holds for the Grad --- Shafranov ansatz \cite{CLV}.}
Most likely, the phenomenon is absent for a general (not axisymmetric) GF,
at least we cannot derive a relation {like} \eqref{6.21} from \eqref{5.12}--\eqref{5.15}.

Studying GFs up to equivalence, we can multiply $u(p,t)$ by a non-vanishing smooth function $\varphi(p)$.
This opportunity was already
used to fix the sign of $u^t$ in \eqref{6.16}. We still have the freedom of multiplying $u(p,t)$ by a positive smooth function $\varphi(p)$ together with the corresponding change of the pressure. Choosing $\varphi(p)=d(p)^{-1/2}$ and denoting the new GF by $(u,p)$ again, we simplify \eqref{6.21} to the following {one}:
\begin{equation}
\vert u(p,t)\vert^2=1.
                                          \label{6.22}
\end{equation}
The inequality \eqref{6.20} becomes now $\vert b(p)\vert<R(p,t)$,
and formula \eqref{6.19} takes the form
\begin{equation}
u^t(p,t)=\frac{\sqrt{R^2(p,t)-b^2(p)}}{R(p,t)},
                                          \label{6.24}
\end{equation}

We continue our analysis under assumptions \eqref{6.16} and \eqref{6.22}.
In virtue of \eqref{6.16}, the equation \eqref{6.12} can be written in the form
$
\frac{\partial(\log u^t)}{\partial t}+\frac{\partial(\log (\vert J\vert R))}{\partial t}=0.
$
This implies
\begin{equation}
u^t(p,t)=\frac{\alpha^2(p)}{\vert J(p,t)\vert R(p,t)}
                                          \label{6.25}
\end{equation}
with some positive smooth function $\alpha(p)$. Comparing \eqref{6.24} and \eqref{6.25}, we arrive to the equation
$\sqrt{R^2(p,t)-b^2(p)}\,\vert J(p,t)\vert=\alpha^2(p)$.

Finally, we simplify the equation \eqref{6.15}. Substituting expressions \eqref{6.18} and \eqref{6.24} for $u^\theta$ and $u^t$ into \eqref{6.15}, we obtain
$\kappa R(R^2-b^2)+b^2Z'_t= R^3\vert J\vert^{-1}$.
Expressing $\vert J\vert$ from \eqref{6.25} and substituting the expression into the latter formula, we arrive to the equation
$
\kappa R(R^2-b^2)+b^2Z'_t=\frac{R^3\sqrt{R^2-b^2}}{\alpha^2}.
$
We have thus proved the following

\begin{theorem} \label{Th6.1}
Let an axisymmetric {Gavrilov flow} $(u,p)$ be defined on an open set $U\subset\{(r,z)\mid r>0\}$.
Assume that every isobaric surface $M_p$ is regular and connected. For the surface of revolution $M_p$, let
$r=R(p,t)>0,\ z=Z(p,t)$ be the arc length parametrization of the generatrix $\Gamma_p$ of $M_p$.
Assume also that in the representation
$u=u^t(p,t)\frac{\partial}{\partial t}+u^\theta(p,t)\frac{\partial}{\partial \theta}$
the function $u^t(p,t)$ does not vanish. Then

{\rm (1)} there exists a smooth positive function $d(p)$ such that
$\vert u\vert^2=d(p)$.
Replacing $(u,p)$ with an equivalent {Gavrilov flow} and denoting the new flow by $(u,p)$ again, we can assume without lost of generality that $u^t(p,t)>0$ and
\begin{equation}
\vert u\vert^2=1.
                                          \label{6.28}
\end{equation}

{\rm (2)} under the assumption \eqref{6.28}, the functions $R(p,t)$ and $Z(p,t)$ satisfy the equations
\begin{eqnarray}
R'_t{}^2+Z'_t{}^2=1,
                                          \label{6.29} \\
\sqrt{R^2-b^2(p)}\,\big\vert R'_pZ'_t-R'_tZ'_p\big\vert=\alpha^2(p),
                                          \label{6.30} \\
\kappa R\big(R^2-b^2(p)\big)+b^2(p)Z'_t=\frac{R^3\sqrt{R^2-b^2(p)}}{\alpha^2(p)}
                                          \label{6.31}
\end{eqnarray}
with some smooth functions $\alpha(p)>0$ and $b(p)$, where $\kappa=\kappa(p,t)$ is the curvature of the plane curve $r=R(p,t),z=Z(p,t)$.
The functions $u^t(p,t)$ and $u^\theta(p,t)$ are expressed through $\big(R(p,t),b(p)\big)$ by
\begin{equation}
u^t=\frac{\sqrt{R^2-b^2(p)}}{R},\quad
u^\theta=\frac{b(p)}{R^2}.
                                          \label{6.32}
\end{equation}
\end{theorem}

Let us make some remarks on Theorem \ref{Th6.1}.

1. First, we attract reader's attention to the hypothesis: $M_p$ are connected surfaces. Otherwise functions $\alpha(p)$ and $b(p)$ can be different on different connected components.

2. Equations \eqref{6.29}--\eqref{6.32} are invariant under some transformations. First, the parameter $t$ is defined up to a shift, i.e., nothing changes after the replacement $R(p,t)=\tilde R\big(p,t+t_0(p)\big),Z(p,t)=\tilde Z\big(p,t+t_0(p)\big)$. Second, the equations are invariant under the changes
$Z(p,t)=\tilde Z(p,t)+z_0$ and $Z(p,t)=-\tilde Z(p,t),\ \kappa=-\tilde\kappa$,
which mean a vertical shift of the origin in ${\mathbb R}^3$ and the change of the direction of the $z$-axis, respectively.

3. Compared {with} Proposition \ref{P5.1}, Theorem \ref{Th6.1} has an important advantage. The unknown functions $(f,u^x,u^y)$ are not separated in the system \eqref{5.12}--\eqref{5.15} and, probably, cannot be separated. On the other hand, equations \eqref{6.29}--\eqref{6.31} involve only the functions $(R,Z)$ that determine isobaric surfaces $M_p$ (the equations involve also $\alpha(p)$ and $b(p)$ that appear as integration constants).
If the system \eqref{6.29}--\eqref{6.31} was solved, the velocity vector field $u$ would be determined by explicit formulas \eqref{6.32}.
Of course, the simplification is possible {due} to the Clairaut integral for the equation of geodesics on a surface of revolution.
Although the Clairaut integral is not mentioned in our proof of Theorem~\ref{Th6.1}, formulas \eqref{6.32} are actually equivalent to the Clairaut integral.


{Recall that physical coordinates $(u_r,u_z,u_\theta)$ of the velocity vector field $u$ are defined by \eqref{6.2}.
The function $u_\theta$ is called {\it the swirl}.
As follows from \eqref{6.3} and \eqref{6.32},
\begin{equation}
(u_r^2+u_z^2)(r,z)=1-\frac{\beta(p)}{r^2},\quad
u_\theta(r,z)=\frac{\sqrt{\beta(p)}}{r}.
                                              \label{8.15}
\end{equation}
}

Theorem \ref{Th6.1} has two other useful forms.

{\bf Case A.} Assume, under hypotheses of Theorem \ref{Th6.1}, that the curve $\Gamma_p$ is the graph of a function $r=f(p,z)$.
Then $f'_p\ne0$ and {equations} \eqref{6.29}--\eqref{6.31} are equivalent to the system
\begin{eqnarray}
{f'_p}^2 = \frac{\alpha(p)(1+{f'_z}^2)}{f^2-b^2(p)},
                                    \label{6.33}\\
\alpha(p) f f''_{zz} = b^2(p) {f'_p}^2 - f^3f'_p(1+{f'_z}^2)
                                    \label{6.34}
\end{eqnarray}
with the same functions $\alpha(p)>0$ and $b(p)$. The function $f$ must satisfy $f(z,p)>\vert b(p)\vert$.
The velocity vector $u$ is now represented as $u=u^z(z,p)\frac{\partial}{\partial z}+u^\theta(z,p)\frac{\partial}{\partial\theta}$, where the functions $u^z$ and $u^\theta$ are expressed through $(f,\alpha,b)$ by
$u^z=\mbox{sgn}(u^z)\,\frac{\sqrt{\alpha(p)}}{f\vert f'_p\vert},\ u^\theta=\frac{b(p)}{f^2}$
{and} $\mbox{sgn}(u^z)=\pm1$ is the sign of $u^z$ that is assumed do not vanish.

{\bf Case B.} Assume, under hypotheses of Theorem \ref{Th6.1}, that the curve $\Gamma_p$ is the graph of a function $z=g(p,r)$.
Then equations \eqref{6.29}--\eqref{6.31} are equivalent to the system
\begin{eqnarray}
\big(r^2-b^2(p)\big)g'_p{}^2-\alpha(p)(1+g'_r{}^2)=0,
                                          \label{6.36}\\
r\alpha(p)g''_{rr}+b^2(p)g'_r g'_p{}^2+r^3 g'_p(1+g'_r{}^2)=0
                                          \label{6.37}
\end{eqnarray}
with the same functions $\alpha(p)>0$ and $b(p)$. The function $g(p,r)$ is considered for $r\in\big(r_1(p),r_2(p)\big)$ with $\vert b(p)\vert<r_1(p)$.
The velocity vector $u$ is now represented as $u=u^r(p,r)\frac{\partial}{\partial r}+u^\theta(p,r)\frac{\partial}{\partial \theta}$, where the functions $u^r$ and $u^\theta$ are expressed through $(g,\alpha,b)$ by
$u^r=\mbox{sgn}(u^r)\,\frac{\sqrt{\alpha(p)}}{r\vert g'_p\vert},\ u^\theta=\frac{b(p)}{r^2}$
{and} $\mbox{sgn}(u^r)=\pm1$ is the sign of $u^r$ that does not~vanish.

\section{Consistency conditions}
\label{sec:cond}

{We will first recall} some basic facts from theory of first order PDEs following \cite[Part~I, Section~14]{Ka2}.
Let us consider the system of two first order PDEs
\begin{equation}
F(x,y,z,p,q)=0,\quad G(x,y,z,p,q)=0,
                                              \label{7.1}
\end{equation}
where $z=z(x,y)$ is an unknown function and $p=z'_x,q=z'_y$. We assume $F$ and $G$ to be sufficiently smooth functions defined for
$(x,y,z)\in U$ and for all $(p,q)\in{\mathbb R}^2$, where $U\subset{\mathbb R}^3$ is an open set. The system \eqref{7.1} is supplied with the initial condition
\begin{equation}
z(x_0,y_0)=z_0
                                              \label{7.2}
\end{equation}
for a point $(x_0,y_0,z_0)\in U$.
The {\it Jacobi brackets} (sometimes also called {\it Mayer brackets}) of functions $F(x,y,z,p,q)$ and $G(x,y,z,p,q)$ are defined by
\begin{equation}
[F,G]=(F'_x+pF'_z)G'_p-(G'_x+p\,G'_z)F'_p+(F'_y+qF'_z)G'_q-(G'_y+qG'_z)F'_q.
                                              \label{7.3}
\end{equation}
The system \eqref{7.1} is said to be an {\it involutory system} if
$[F,G]\equiv0$ for $(x,y;z,p,q)\in U\times{\mathbb R}^2$.
The system \eqref{7.1} is said to be a {\it complete system on the open set} $U\times{\mathbb R}^2$ if the equation
\begin{equation}
[F,G]=0
                                              \label{7.4}
\end{equation}
is an algebraic corollary of the system \eqref{7.1}, i.e., if \eqref{7.4} holds
for $(x,y;z,p,q)\in U\times{\mathbb R}^2$ satisfying \eqref{7.1}. The equation \eqref{7.4} is called the {\it consistency condition} (or {\it integrability condition}) for the system \eqref{7.1}.
In the case of a complete system, for an arbitrary point $(x_0,y_0,z_0)\in U$, the initial value problem \eqref{7.1}--\eqref{7.2} has a unique solution at least in some neighborhood of the point $(x_0,y_0)$. Several methods are known for the numerical solution {of} the IVP \eqref{7.1}--\eqref{7.2}, the Mayer method is the most popular one \cite{Ka2}.

\medskip

{We} return to axisymmetric GFs.
Under certain additional conditions, systems \eqref{6.33}--\eqref{6.34} and \eqref{6.36}--\eqref{6.37} are equivalent. We study the system \eqref{6.33}--\eqref{6.34} as the simplest one. Recall that the system is considered in a neighborhood of a regular point, where the transform $(p,z)\mapsto\big(f(p,z),z\big)$ is one-to-one. Therefore the derivative $f'_p$ does not vanish.
Note that only $b^2$ is involved in \eqref{6.33}--\eqref{6.34}, not the function $b$ itself.
To simplify our formulas a {little} bit, we introduce the function $\beta(p)=b^2(p)\ge0$ and rewrite the system \eqref{6.33}--\eqref{6.34} as
\begin{eqnarray}
(f^2-\beta){f'_p}^2 -\alpha(1+{f'_z}^2)=0,
                                    \label{7.5}\\
\alpha f f''_{zz} -\beta {f'_p}^2 + f^3f'_p(1+{f'_z}^2)=0.
                                    \label{7.6}
\end{eqnarray}
{The function $f$ is assumed to satisfy} the inequality
\begin{equation}
f(p,z)>\sqrt{\beta(p)}.
                                    \label{7.7}
\end{equation}
Given functions $\alpha(p)>0$ and $\beta(p)\ge0$, \eqref{7.5}--\eqref{7.6} is an overdetermined system of two PDEs in one unknown function $f(p,z)$.
The overdeterminess is caused by the circumstance mentioned in Introduction: a GF is defined by the overdetermined system \eqref{1.1}--\eqref{1.3}.
We pose the question: \textit{What conditions should be imposed on $\big(\alpha(p),\beta(p)\big)$ for solvability of the system \eqref{7.5}--\eqref{7.6} at least locally, i.e., in a neighborhood of a given point} $(p_0,z_0)$?
For a fixed $p$, \eqref{7.6} can be considered as a second order ODE with an unknown function $f_p(z)=f(p,z)$.
{Observe} that the variable $z$ does not explicitly participate in \eqref{7.6}.
{As well known \cite[Ssection~15.3]{Ka}, such} an equation can be reduced to a first order ODE.
{The} observation is realized by the following statement.

\begin{lemma} \label{L7.1}
Let $C^1$-functions $\alpha(p)>0$ and $\beta(p)\ge0$ be defined on an interval $(p_1,p_2)$, where $-\infty\le p_1<p_2\le\infty$.
Then the following {statements}
are valid.

{\rm (1)} Let $f(p,z)$ be a solution to the system \eqref{7.5}--\eqref{7.6} on a rectangle
\begin{equation}
(p,z)\in(p_1,p_2)\times(z_1,z_2),
                                    \label{7.8}
\end{equation}
and let the inequality \eqref{7.7} be valid on the rectangle.
Then there exists a $C^1$-function $\gamma(p)$ on the interval $(p_1,p_2)$ satisfying the equation
\begin{equation}
f^2(\varepsilon f^2+\gamma)^2(1+f'_z{}^2)-4\alpha(f^2-\beta)=0
                                    \label{7.9}
\end{equation}
and the inequalities
\begin{equation}
0 < \varepsilon f^2+\gamma \le \frac{2\alpha^{1/2}\sqrt{f^2 - \beta}}{f},
                                    \label{7.10}
\end{equation}
where $\varepsilon=\pm1$ is the sign of $f'_p$ that does not vanish.

{\rm (2)} Conversely, let
$f(p,z)>0$ and $\gamma(p)$ satisfy \eqref{7.5} and \eqref{7.9}--\eqref{7.10}. Assume additionally that, for every $p\in(p_1,p_2)$, the derivative $f'_z(p,z)$ is not identically equal to zero on any interval $(z'_1,z'_2)\subset(z_1,z_2)$. Then $f$ solves
\eqref{7.5}--\eqref{7.6} on the rectangle \eqref{7.8}.
\end{lemma}

\begin{proof}
{First of all} we find from \eqref{7.5}
\begin{equation}
f'_p =\varepsilon\, \frac{\alpha^{1/2}\sqrt{1+{f'_z}^2}}{\sqrt{f^2-\beta}},
                                    \label{7.11}
\end{equation}
where $\varepsilon=\pm1$ is the sign of $f'_p$,
and substitute the expression into \eqref{7.6}
\begin{equation}
f''_{zz} = \frac{\beta(1+{f'_z}^2)}{f(f^2-\beta)} -\varepsilon\, \frac{f^2(1+{f'_z}^2)^{3/2}}{\alpha^{1/2}\sqrt{f^2-\beta}}.
                                    \label{7.12}
\end{equation}
Let us show that equations \eqref{7.11} and \eqref{7.12} imply
\begin{equation}
\frac{\partial}{\partial z}\Big(\frac{\sqrt{f^2-\beta}}{f\sqrt{1+{f'_z}^2}} -\varepsilon\,\frac{f^2}{2\alpha^{1/2}}\Big)=0.
                                    \label{7.13}
\end{equation}
{To this end} we implement the differentiation in \eqref{7.13}. The result can be written as
\begin{equation}
\Big(f''_{zz}-\frac{\beta(1+{f'_z}^2)}{f(f^2-\beta)}+\varepsilon\,\frac{f^2(1+{f'_z}^2)^{3/2}}{\alpha^{1/2}\sqrt{f^2-\beta}}\Big)f'_z=0.
                                    \label{7.14}
\end{equation}
By \eqref{7.12}, the left-hand side of \eqref{7.14} is identically zero. This proves \eqref{7.13}.
The equation \eqref{7.13} means the existence of a function $\gamma(p)$ such that
$
\frac{\sqrt{f^2-\beta}}{f\sqrt{1+{f'_z}^2}} -\varepsilon\,\frac{f^2}{2\alpha^{1/2}}=\frac{\gamma}{2\alpha^{1/2}}.
$
This can be written in the form
\begin{equation}
2\alpha^{1/2}\sqrt{f^2-\beta}=f\sqrt{1+{f'_z}^2}(\varepsilon f^2+\gamma).
                                    \label{7.15}
\end{equation}
By \eqref{7.7}, $f>0$ and $f^2-\beta>0$. Therefore \eqref{7.15} implies the inequalities \eqref{7.10}.
Squaring the equation \eqref{7.15}, we get \eqref{7.9}. We have proved the first statement of the lemma.

The second {statement} of the lemma is proved by {reversing presented arguments with} the following additional remark.
To pass from \eqref{7.14} to \eqref{7.12}, we need to remove the factor $f'_z$ on the left-hand side of \eqref{7.14}.
To do this, it suffices to assume that for every $p\in(p_1,p_2)$, the derivative $f'_z(p,z)$ is not identically zero on any interval $(z'_1,z'_2)\subset(z_1,z_2)$.
\end{proof}

By Lemma \ref{L7.1}, the system \eqref{7.5}--\eqref{7.6} {is equivalent to} the following {one}:
\begin{equation}
\begin{aligned}
(f^2-\beta){f'_p}^2 -\alpha(1+{f'_z}^2)&=0,\\
f^2(\varepsilon f^2+\gamma)^2(1+f'_z{}^2)-4\alpha(f^2-\beta)&=0.
\end{aligned}
                                             \label{7.16}
\end{equation}
Introducing the notations $\pi=f'_p,\quad\zeta=f'_z$, we write the system as
\begin{equation}
F(p,z,f,\pi,\zeta)=0,\quad G(p,z,f,\pi,\zeta)=0,
                                              \label{7.17}
\end{equation}
where
\begin{equation}
\begin{aligned}
F(p,z,f,\pi,\zeta)&=(f^2-\beta)\pi^2-\alpha(\zeta^2+1),\\
G(p,z,f,\pi,\zeta)&=f^2(\varepsilon f^2+\gamma)^2(\zeta^2+1)-4\alpha(f^2-\beta).
\end{aligned}
                                              \label{7.18}
\end{equation}
Assume that
$\alpha(p)>0,\, \beta(p)\ge0$ and $\gamma(p)$ are defined and smooth on an interval $(p_1,p_2)$,
where $-\infty\le p_1<p_2\le\infty$.
Then
$F$ and $G$ are defined and smooth in $U\times{\mathbb R}^2$, where
\begin{equation}
U=\big\{(p,z,f)\in{\mathbb R}^3\vert \ p_1<p<p_2,\ f>\!\sqrt{\beta},\ 0<\varepsilon f^2+\gamma<2\sqrt{\alpha}\sqrt{f^2-\beta}/f\big\}.
                                              \label{7.19}
\end{equation}
The functions $F$ and $G$ are actually independent of $z$ and depend on $p$ through the functions $\alpha(p),\beta(p),\gamma(p)$ only.
Up to notations, the system \eqref{7.17} is of the form \eqref{7.1}.

\begin{theorem} \label{Th7.1}
Given $C^1$-functions $\alpha(p)>0,\beta(p)\ge0$ and $\gamma(p)$ on an interval $(p_1,p_2)$, define $F$ and $G$ by \eqref{7.18} and consider the system of PDEs \eqref{7.17}, where $f=f(p,z)$ is an unknown function and $\pi=f'_p,\zeta=f'_z$. Define an open set $U\subset{\mathbb R}^3$ by \eqref{7.19}.
Then

{\rm 1.} The system \eqref{7.17} is complete
on $U\times{\mathbb R}^2$ if the functions $\alpha,\beta,\gamma$ satisfy
\begin{equation}
\alpha(p)=\alpha_0e^{3 p}\quad\mbox{\rm with some constant}\ \alpha_0>0,
                                    \label{7.20}
\end{equation}
\begin{equation}
\beta'+2\varepsilon\gamma'+3\beta+\varepsilon\gamma=0,
                                             \label{7.21}
\end{equation}
\begin{equation}
\gamma\beta'-2\beta\gamma'+3\beta\gamma-4\varepsilon\alpha+\varepsilon\gamma^2=0,
                                             \label{7.22}
\end{equation}
where either $\varepsilon=1$ or $\varepsilon=-1$ and $\beta+\varepsilon\gamma\ne0$.
But \eqref{7.17} is not an involutory system.

{\rm 2.} Conversely, assume \eqref{7.17} {to be} a complete system on $U\times{\mathbb R}^2$, where the open set $U\subset{\mathbb R}^3$ is defined by \eqref{7.19}. Assume additionally that for every $p_0\in(p_1,p_2)$ there exists $z_0$ such that \eqref{7.17} has a solution $f(p,z)$ in a neighborhood of $(p_0,z_0)$ satisfying $f'_z(p_0,z_0)\ne0$. Then the functions $\alpha,\beta,\gamma$ satisfy \eqref{7.20}--\eqref{7.22} with some $\varepsilon=\pm1$.
\end{theorem}

\begin{proof}
To agree \eqref{7.1} and \eqref{7.17}, we need to change the variables in \eqref{7.1} as follows:
$x:=p,\ y:=z,\ z:=f,\ p:=\pi$ and $q:=\zeta$.
Then the formula \eqref{7.3} takes the form
\begin{equation}
[F,G]=(F'_p+\pi F'_f)G'_\pi-(G'_p+\pi G'_f)F'_\pi+(F'_z+\zeta F'_f)G'_\zeta-(G'_z+\zeta G'_f)F'_\zeta.
                                              \label{7.23}
\end{equation}
We find the derivatives by differentiating \eqref{7.18}:
\begin{equation}
\begin{array}{ll}
F'_p=-\beta'\pi^2{-}\alpha'(\zeta^2+1),&
G'_p=2\gamma'f^2(\varepsilon f^2{+}\gamma)(\zeta^2{+}1)-4\alpha'(f^2{-}\beta)+4\alpha\beta',\\
[5pt]
F'_z=0,&
G'_z=0,\\
[5pt]
F'_f=2f\pi^2,&
G'_f=2f(\varepsilon f^2{+}\gamma)^2(\zeta^2{+}1){+}4\varepsilon f^3(\varepsilon f^2{+}\gamma)(\zeta^2{+}1){-}8\alpha f,\\
[5pt]
F'_\pi=2(f^2-\beta)\pi,&
G'_\pi=0,\\
[5pt]
F'_\zeta=-2\alpha\zeta,&
G'_\zeta=2f^2(\varepsilon f^2+\gamma)^2\zeta.
\end{array}
                                              \label{7.24}
\end{equation}
Since $G'_\pi=F'_z=G'_z=0$, the formula \eqref{7.23} simplifies to the following {one}:
\begin{equation}
[F,G]=-G'_pF'_\pi-\pi G'_fF'_\pi+\zeta F'_fG'_\zeta-\zeta G'_fF'_\zeta.
                                              \label{7.25}
\end{equation}
Substituting values \eqref{7.24} into \eqref{7.25}, we obtain
\begin{equation}
\begin{aligned}
&\frac{1}{4}[F,G]=-\pi(f^2-\beta)\big(\gamma'f^2(\varepsilon f^2+\gamma)(\zeta^2+1)-2\alpha'(f^2-\beta)+2\alpha\beta'\big)\\
&-\pi^2f(f^2-\beta)\big((\varepsilon f^2+\gamma)^2(\zeta^2+1)+2\varepsilon f^2(\varepsilon f^2+\gamma)(\zeta^2+1)-4\alpha\big)\\
&+\pi^2\zeta^2f^3(\varepsilon f^2+\gamma)^2
+\zeta^2\alpha f\big((\varepsilon f^2+\gamma)^2(\zeta^2+1)+2\varepsilon f^2(\varepsilon f^2+\gamma)(\zeta^2+1)-4\alpha\big).
\end{aligned}
                                              \label{7.26}
\end{equation}
The right-hand side of \eqref{7.26} is a 7th degree polynomial in $f$ and the coefficient at $f^7$ is $-\pi^2(2\zeta^2+3)\neq0$.
Thus, \eqref{7.17} is not an involutory system.
Now, we prove that \eqref{7.17} is a complete system. {To this end} we derive from \eqref{7.17}--\eqref{7.18}
\begin{equation}
\zeta^2=\frac{4\alpha(f^2-\beta)}{f^2(\varepsilon f^2+\gamma)^2}-1
                                              \label{7.27}
\end{equation}
and
$\pi=\varepsilon\,\frac{\alpha^{1/2}\sqrt{\zeta^2+1}}{\sqrt{f^2-\beta}}$,
where $\varepsilon=\pm1$ is the sign of $\pi$. We find from two last equalities
\begin{equation}
\pi=\frac{2\varepsilon\alpha}{f(\varepsilon f^2+\gamma)}.
                                              \label{7.28}
\end{equation}
Substituting expressions \eqref{7.27}--\eqref{7.28} into \eqref{7.26}, we obtain
\begin{equation}
\begin{aligned}
\frac{1}{16\alpha^2}[F,G]&=-\varepsilon\frac{2(f^2-\beta)}{f(\varepsilon f^2+\gamma)}
\big(\frac{\gamma'(f^2-\beta)}{\varepsilon f^2+\gamma}-\frac{\alpha'}{2\alpha}(f^2-\beta)+\frac{\beta'}{2}\big)\\
&-\frac{4\alpha(f^2-\beta)}{f(\varepsilon f^2+\gamma)^2}\big(\frac{f^2-\beta}{f^2}
+2\varepsilon \,\frac{f^2-\beta}{\varepsilon f^2+\gamma}-1\big)
+f\big(\frac{4\alpha(f^2-\beta)}{f^2(\varepsilon f^2+\gamma)^2}-1\big)\\
&+f\big(\frac{4\alpha(f^2-\beta)}{f^2(\varepsilon f^2+\gamma)^2}-1\big)
\big(\frac{f^2-\beta}{f^2}+2\varepsilon \,\frac{f^2-\beta}{\varepsilon f^2+\gamma}-1\big).
\end{aligned}
                                              \label{7.29}
\end{equation}
We are interested in the case when $[F,G]=0$. Equating the right-hand side of \eqref{7.29} to zero
and multiplying the resulting equality by $f^3(\varepsilon f^2+\gamma)^3$, we arrive to the equation
\begin{equation}
\begin{aligned}
&-2\varepsilon f^2(f^2-\beta)(\varepsilon f^2+\gamma)
\big(\gamma'(f^2-\beta)-\frac{\alpha'}{2\alpha}(f^2-\beta)(\varepsilon f^2+\gamma)+\frac{\beta'}{2}(\varepsilon f^2+\gamma)\big)\\
&-4\alpha(f^2-\beta)\big((f^2-\beta)(\varepsilon f^2+\gamma)+2\varepsilon f^2(f^2-\beta)-f^2(\varepsilon f^2+\gamma)\big)\\
&+4\alpha f^2(f^2-\beta)(\varepsilon f^2+\gamma)-f^4(\varepsilon f^2+\gamma)^3\\
&+\big(4\alpha(f^2{-}\beta)-f^2(\varepsilon f^2{+}\gamma)^2\big)
\big((f^2{-}\beta)(\varepsilon f^2{+}\gamma)+2\varepsilon f^2(f^2{-}\beta)-f^2(\varepsilon f^2{+}\gamma)\big)=0.
\end{aligned}
                                              \label{7.30}
\end{equation}
The left-hand side of the equation \eqref{7.30} is a polynomial of 10th degree in $f$.
It is almost unbelievable, but the degree of the polynomial can be decreased to 4. Namely, the equation \eqref{7.30} is equivalent to the following one:
\begin{equation}
\begin{aligned}
-f^2(f^2-\beta)(\varepsilon f^2+\gamma)
\big[&\big(3-\frac{\alpha'}{\alpha}\big)f^4
+\big(2\varepsilon\gamma'-\varepsilon\frac{\alpha'}{\alpha}\gamma+\frac{\alpha'}{\alpha}\beta+\beta'+4\varepsilon\gamma\big)f^2\\
&+\big(-2\varepsilon\beta\gamma'+\varepsilon\frac{\alpha'}{\alpha}\beta\gamma+\varepsilon\gamma\beta'-4\alpha+\gamma^2\big)\big]=0.
\end{aligned}
                                              \label{7.31}
\end{equation}
Indeed, a simple (though rather cumbersome) calculation {shows} that the polynomials on the left-hand sides of \eqref{7.30} and \eqref{7.31} are {identically} equal.
 By \eqref{7.19}, the factor $f^2(f^2-\beta)(\varepsilon f^2+\gamma)$ does not vanish on $U$.
Therefore the equation \eqref{7.31} is equivalent to the following {one}:
\begin{equation}
\big(3-\frac{\alpha'}{\alpha}\big)f^4
+\big(2\varepsilon\gamma'-\varepsilon\frac{\alpha'}{\alpha}\gamma+\frac{\alpha'}{\alpha}\beta+\beta'+4\varepsilon\gamma\big)f^2
+\big(-2\varepsilon\beta\gamma'+\varepsilon\frac{\alpha'}{\alpha}\beta\gamma+\varepsilon\gamma\beta'-4\alpha+\gamma^2\big)=0.
                                              \label{7.32}
\end{equation}
Equating coefficients of the polynomial on the left-hand side of \eqref{7.32} to zero, we arrive to the system of ODEs
\begin{equation}
\begin{aligned}
&\alpha'-3\alpha=0,\\
&\beta'+2\varepsilon\gamma'+3\beta+\varepsilon\gamma=0,\\
&\varepsilon\gamma\beta'-2\varepsilon\beta\gamma'+3\varepsilon\beta\gamma-4\alpha+\gamma^2=0.
\end{aligned}
                                              \label{7.33}
\end{equation}
{This is equivalent to \eqref{7.20}--\eqref{7.22}.}

{It remains to discuss the passage from \eqref{7.32} to \eqref{7.33}.}
Of course, \eqref{7.33} implies \eqref{7.32}.
This proves the first {statement} of Theorem~\ref{Th7.1}.
To prove that \eqref{7.32} implies \eqref{7.33}, we need for each $p_0\in(p_1,p_2)$ to have at least three distinct $z_1,z_2,z_3$ such that the values $f^2(p_0,z_1),\,f^2(p_0,z_2),\,f^2(p_0,z_3)$ are pairwise different and \eqref{7.32} holds at $(p_0,z_1),(p_0,z_2)$ and $(p_0,z_3)$. The existence of such $z_1,z_2,z_3$ is guaranteed by the hypothesis of the second assertion of Theorem~\ref{Th7.1}: For each $p_0\in(p_1,p_2)$, there exists $z_0$ such that the system \eqref{7.17} has a solution $f(p,z)$ in a neighborhood of $(p_0,z_0)$ satisfying $f'_z(p_0,z_0)\ne0$.
\end{proof}

{\bf Remark.} Roughly speaking, Theorem \ref{Th7.1} means that the relations \eqref{7.20}--\eqref{7.22} constitute the consistency condition for the system \eqref{7.17}. Nevertheless, we emphasize that two statements of Theorem \ref{Th7.1} are not exactly converse to each other. For example, the axisymmetric GF  \eqref{3.2} (isobaric surfaces are cylinders and particle trajectories are either circles or spiral lines) corresponds to the solution $f(p,z)=2p^{1/2}$ to the system \eqref{7.16} with $(p_1,p_2)=(0,\infty)$ and
\begin{equation}
\alpha(p)=4a_0,\ \ \beta(p)=4(1-a_0)p,\ \ \gamma(p)=4(a_0-p)\ \ (a_0=\mbox{const},\,0<a_0\le1).
                                              \label{7.34}
\end{equation}
The functions \eqref{7.34} do not satisfy the consistency conditions \eqref{7.20}--\eqref{7.22}.
For this solution, $f'_z\equiv0$ and the second statement of Theorem \ref{Th7.1} does not apply.


An analog of Lemma~\ref{L7.1} is valid for the system \eqref{6.36}--\eqref{6.37} with minor changes (the variables $(p,z)$ are replaced by $(p,r)$, the function $f$ is replaced by $g$, etc.). The~cor\-responding system of two first order PDEs looks as follows:
\begin{equation}
\begin{aligned}
(r^2-\beta)g'_p{}^2-\alpha(1+g'_r{}^2)&=0,\\
r^2(\gamma-\tau r^2)^2(1+g'_r{}^2)-4\alpha(r^2-\beta)g'_r{}^2&=0,
\end{aligned}
                                          \label{7.35}
\end{equation}
where $\alpha(p)>0,\beta(p)\ge0$ and $\gamma(p)$ are the same functions as in \eqref{7.16} and $\tau=\pm1$ is the sign of $g'_p$ that does not vanish. An analog of Theorem \ref{Th7.1} is valid for the system \eqref{7.35} with the same consistency conditions \eqref{7.20}--\eqref{7.22}.

\section{Local and $z$-periodic axisymmetric Gavrilov flows}
\label{sec:2ex}

{Here,
we discuss the numerical method for constructing axisymmetric GFs on the base of the system \eqref{7.16}.}
For the initial condition
$f(p_0,z_0)=f_0$, we can assume without lost of generality that $p_0=z_0=0$ since the system \eqref{7.16} is invariant under the change $p\to p+\mbox{const}$ and $z\to z+\mbox{const}$. Thus, the initial condition for $f$ is
\begin{equation}
f(0,0)=f_0.
                                             \label{9.1}
\end{equation}
First, we will find the functions $\alpha(p),\beta(p),\gamma(p)$. The function $\alpha$ is given explicitly by \eqref{7.20} 
with an arbitrary constant $\alpha_0=\alpha(0)>0$. The functions $\beta(p)\ge0$ and $\gamma(p)$ solve the system \eqref{7.21}--\eqref{7.22} supplied with the initial conditions
\begin{equation}
\beta(0)=\beta_0,\quad \gamma(0)=\gamma_0.
                                             \label{9.2}
\end{equation}
In particular, $\beta_0\ge0$.
For {the possibility} to write the system \eqref{7.21}--\eqref{7.22} in the form $\beta'=B(\beta,\gamma),\ \gamma'=\Gamma(\beta,\gamma)$,
we have to impose the restrictions
\begin{equation}
\beta\ge0,\quad \beta+\varepsilon\gamma\neq0.
                                             \label{9.3}
\end{equation}
In particular, $\beta_0$ and $\gamma_0$ must satisfy these inequalities.
Let $(p_1,p_2)$ be the {maximal interval} such that the solution to the Cauchy problem \eqref{7.21}--\eqref{7.22}, \eqref{9.2} exists on $(p_1,p_2)$ and satisfies \eqref{9.3}. Here, $p_1=p_1(\alpha_0,\beta_0,\gamma_0)<0<p_2=p_2(\alpha_0,\beta_0,\gamma_0)$.
By Theorem \ref{Th7.1}, the system \eqref{7.16} is complete on $U\times{\mathbb R}^2$, where $U\subset{\mathbb R}^3$ is defined by \eqref{7.19}. General theory \cite[Chapter~1, Section~14]{Ka} guarantees the existence and uniqueness of a solution to the IVP \eqref{7.16}, \eqref{9.1} for any $(0,0,f_0)\in U$ at least in some neighborhood of $(p_0,z_0)=(0,0)$.

Thus, a solution $f(p,z)$ to the IVP \eqref{7.16}, \eqref{9.1} exists in some neighborhood of $(0,0)$ and is uniquely determined by 5 constants $(\alpha_0,\beta_0,\gamma_0,f_0,\varepsilon)$
chosen so~that
\begin{equation}
\begin{aligned}
&\alpha_0>0,\quad \beta_0\ge0,\quad \beta_0+\varepsilon\gamma_0\neq0,\quad f_0>0,\quad\varepsilon=\pm1,\\
&f_0>\beta_0^{1/2},\quad 0<\varepsilon f_0^2+\gamma_0<2\alpha_0^{1/2}f_0^{-1}\sqrt{f_0^2-\beta_0}.
\end{aligned}
                                             \label{9.4}
\end{equation}
The inequalities on the second line of \eqref{9.4} come from \eqref{7.19}. Let us denote this unique solution by $f(p,z;\alpha_0,\beta_0,\gamma_0,f_0,\varepsilon)$.

In {the general case}, $f(p,z;\alpha_0,\beta_0,\gamma_0,f_0,\varepsilon)$ is a {\it local solution}, i.e., is defined in some neighborhood $U(\alpha_0,\beta_0,\gamma_0,f_0,\varepsilon)\subset{\mathbb R}^2$ of the point $(0,0)$. Nevertheless, for some values of the parameters, it can happen that $(p'_1,p'_2)\times{\mathbb R}\subset U(\alpha_0,\beta_0,\gamma_0,f_0,\varepsilon)$ for some $-\infty\le p'_1<0<p'_2\le\infty$, in such a case we speak on a {\it global solution} defined on $(p'_1,p'_2)\times{\mathbb R}$. Global solutions are of particular interest.
Unfortunately, so far we have neither necessary nor sufficient conditions on the parameters $(\alpha_0,\beta_0,\gamma_0,f_0,\varepsilon)$ for the existence of a global solution.

Everything said above in this section is valid for the numerical method based on the system \eqref{7.35}.

Global solutions most often appear due to periodicity {with the help of} the following

\begin{lemma}\label{L9.1}
Let a solution $f(p,z)$ of the system \eqref{7.16} be defined on a rectangle $(p_1,p_2)\times(z_1,z_2)$ and satisfy \eqref{7.7}.
Assume the existence of $(p_0,z_0)\in(p_1,p_2)\times(z_1,z_2)$ such that $f'_z(p_0,z_0)=0$ and $f''_{zz}(p_0,z_0)\ne0$.
Then, at least in some neighborhood of $p_0$, the solution is symmetric with respect to $z_0$, i.e.,
\begin{equation}
f(p,z)=f(p,-z+2z_0)\quad\mbox{\rm for}\ p\in(p_0-\delta,p_0+\delta)\ \mbox{\rm with some}\ \delta>0.
                                             \label{9.5}
\end{equation}
Hence, $f$ can be extended to a solution defined in $(p_0-\delta,p_0+\delta)\times(z_0-\Delta,z_0+\Delta)$ such that $(z_1,z_2)\subset(z_0-\Delta,z_0+\Delta)$.
\end{lemma}

\begin{proof}
By the implicit function theorem, there exists a smooth function $z=\zeta(p)$ defined for $p\in(p_0-\delta,p_0+\delta)$ with some $\delta>0$ such that $\zeta(p_0)=z_0$ and
\begin{equation}
f'_z(p,\zeta(p))=0.
                                             \label{9.6}
\end{equation}
We are going to prove that $\zeta(p)$ is actually a constant function. To this end we differentiate~\eqref{9.6}:
 $f''_{pz}(p,\zeta(p))+f''_{zz}(p,\zeta(p))\,\zeta'(p)=0$.
By choosing a smaller $\delta$, we can assume that $f''_{zz}(p,\zeta(p))\neq0$. Thus, to prove the equality $\zeta'=0$, we have to demonstrate that $f''_{pz}(p,\zeta(p))=0$. To this end we differentiate the first equation of the system \eqref{7.16} with respect to $z$ (recall that $\alpha$ and $\beta$ are independent of $z$)
$
(f^2-\beta)f'_pf''_{pz}+f(f'_p)^2f'_z-\alpha f'_zf''_{zz}=0.
$
Setting $z=\zeta(p)$ here and using \eqref{9.6}, we obtain $(f^2-\beta)f'_pf''_{pz}\vert_{z=\zeta(p)}=0$.
Since $f'_p\neq0$ and $f^2-\beta>0$, this implies $f''_{pz}(p,\zeta(p))=0$. We have thus proved
\begin{equation}
f'_z(p,z_0)=0\quad\mbox{for}\ p\in(p_0-\delta,p_0+\delta).
                                             \label{9.8}
\end{equation}
The system \eqref{7.16} is invariant under the change $z\to 2z_0-z$. {Therefore} \eqref{9.8} implies~\eqref{9.5}.
\end{proof}

Now, under hypotheses of Lemma \ref{L9.1}, assume the existence of a second point
$(p'_0,z'_0)\in(p_1,p_2)\times(z_1,z_2)\ (z_0\neq z'_0)$ such that $f'_z(p'_0,z'_0)=0$ and $f''_{zz}(p'_0,z'_0)\neq0$. An analog of \eqref{9.5} holds for $z'_0$ with some $\delta'>0$. Assume additionally that
$$
({\tilde p}_1,{\tilde p}_2)=(p_0-\delta,p_0+\delta)\cap(p'_0-\delta',p'_0+\delta')\neq\emptyset.
$$
Using symmetries with respect to $z_0$ and $z'_0$, we extend $f(p,z)$ to a global $z$-periodic solution defined on $({\tilde p}_1,{\tilde p}_2)\times{\mathbb R}$.
The period is equal to $\vert z_0-z'_0\vert$ if $p_0=p'_0$, otherwise the period is equal to $2\vert z_0-z'_0\vert$.

\smallskip
An example of $z$-periodic GF is presented on Figure~\ref{fig:F1}. The solution
$
f(p,z)=f(p,z;$ $\alpha_0,\beta_0,\gamma_0,f_0,\varepsilon)
$
to the IVP \eqref{7.16}, \eqref{9.1} was computed for $\alpha_0=1,\beta_0=0.01, \gamma_0=0.5, f_0=0.97,\varepsilon=1$.
The solution to the Cauchy problem \eqref{7.21}--\eqref{7.22}, \eqref{9.2} exists for $-0.07<p<7.48$.
For $0<p< p_c\approx 0.25$, the flow turns out to be $z$-periodic with the period
approximately equal to
$2.8$. Graphs $\Gamma_i$ of functions $r=f(p_i,\,z)$ for
$p_i=0.02\,i$ and $i=0,\dots,8$ are drawn on the left-hand side of Figure~\ref{fig:F1}.
Each curve $\Gamma_i$ is the generatrix of the isobaric surface $M_{p_i}$.
The velocity vector field $u$ is drawn on the right picture for $\Gamma_0$ and $\Gamma_8$. We used \eqref{8.15} for computing physical coordinates of $u$.
The two-dimensional vector field $u_re_r+u_ze_z$ is tangent to $\Gamma_0$ and $\Gamma_8$, and the vector field $u_\theta e_\theta$ (the swirl) is orthogonal to the plane of the picture. The latter vector field is drawn by vectors orthogonal to $\Gamma_0$ and $\Gamma_8$ in order to avoid 3D pictures.
In order to make a nice picture, both fields $u_re_r+u_ze_z$ and $u_\theta e_\theta$ are drawn in the scale $1:4$, i.e., $\vert u\vert=0.25$ on the picture. But we remember that actually~$\vert u\vert=1$.

{As we have mentioned in the remark written after the proof of Theorem~\ref{Th7.1}, the axisymmetric GF \eqref{3.2} is an exception in a certain sense: it cannot be obtained in the scope of Theorem \ref{Th7.1}. We can state now the important phenomenon: the simplest axisymmetric GF \eqref{3.2} constitutes the main obstacle for our numerical method.
As is seen from the Figure~\ref{fig:F1}, for $p$ close to the critical value $p_c$, the flow is close to the $z$-independent solution \eqref{3.2}: isobaric surfaces are close to circular cylinders and particle trajectories are close to spiral lines intersecting parallels approximately at the angle of $\pi/4$. The hypothesis $f''_{zz}(p_0,z_0)\neq0$ of Lemma~\ref{L9.1} is violated at $p_0=p_c$. All our numerics become very unstable when $p$ approaches $p_c$. The flow shown on Figure~\ref{fig:F1} exists for $p>p_c$ as well, but it does not need to be periodic with the same period $\approx2.8$ for $p>p_c$.}

\begin{figure}
\noindent\includegraphics[scale=0.46]
{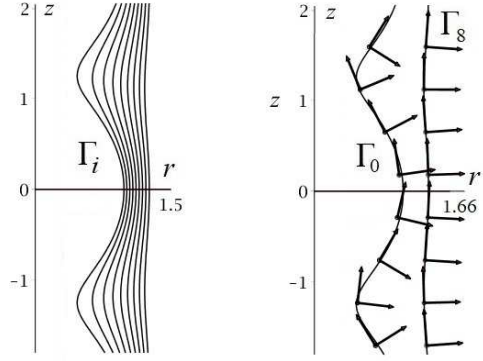}
\caption{\small
Periodic axisymmetric GF.
}
\label{fig:F1}
\end{figure}

\section{Structure of an axisymmetric {Gavrilov flow} in a neighborhood of a~minimum point of the pressure}
\label{sec:structure}

In \cite{G}, the existence of a pair $(\tilde u,\tilde p)$ is proved such that
(a)~$(1,0)$ is a non-degenerate minimum point of the function
$\tilde p=\psi(r,z)\in C^\infty(U)$ and $\psi(1,0)=0$, where $U\subset\{(r,z)\mid r>0\}$ is a neighborhood of $(1,0)$;
(b)~$(\tilde u,\tilde p)$ is an axisymmetric GF in $U\setminus\{(1,0)\}$;
(c)~the split Bernoulli law for the flow is of the simplest form
\begin{equation}
\vert\tilde u\vert^2=3\tilde p.
                                             \label{8.1}
\end{equation}
The flow $(\tilde u,\tilde p)$ does not satisfy the normalization condition \eqref{6.28}. To apply our equations, we must replace
$(\tilde u,\tilde p)$ with an equivalent GF $(u,p)$ such that $\vert u\vert^2=1$.
By the definition \eqref{1.4} of equivalent GFs,
$u=\varphi(\tilde p)\tilde u, \ {\rm grad}\,p=\varphi^2(\tilde p){\rm grad}\,\tilde p$
with some non-vanishing function $\varphi(\tilde p)$. As is seen from \eqref{8.1},
$\varphi^2(\tilde p)=1/\vert\tilde u\vert^2=1/3\,\tilde p$.~Thus,
\begin{equation}
p(r,z)=\frac{1}{3}\ln\psi(r,z)+C\quad(C=\mbox{const}),
                      \label{8.2}
\end{equation}
Hence, the pressure function gets a logarithmic singularity at the point $(1,0)$ after normalization.
Comparing \eqref{7.20} and \eqref{8.2}, we conclude that
\begin{equation}
\alpha=c\,\psi\quad(c=\mbox{const}>0).
                                              \label{8.3}
\end{equation}
Since $p\to-\infty$ as $(r,z)\to(1,0)$, it is natural to choose initial conditions for the system \eqref{7.21}--\eqref{7.22} at $p=-\infty$. The conditions are
\begin{equation}
\beta(-\infty)=\lim\limits_{p\to-\infty}\beta(p)=1/3,\quad
\gamma(-\infty)=\lim\limits_{p\to-\infty}\gamma(p)=-1.
                      \label{8.4}
\end{equation}
Indeed, for $p$ ``close" to $-\infty$, $\Gamma_p$ is a ``small" closed curve around $(1,0)$.
The tangent line to $\Gamma_p$ is vertical at some point
$(r_1,z_1)=\big(r_1(p),z_1(p)\big)\in \Gamma_p$. In a neighborhood of $(r_1,z_1)$, the curve $\Gamma_p$ is the graph of a function $r=f(p,z)$ solving the system \eqref{7.16} and satisfying
$f'_z(p,z_1)=0$. Setting $z=z_1$ in the second equation of \eqref{7.16}, we get
\begin{equation}
f^2(p,z_1)\big(\varepsilon f^2(p,z_1)+\gamma(p)\big)^2-4\alpha(p)\big(f^2(p,z_1)-\beta(p)\big)=0.
                      \label{8.5}
\end{equation}
In view of our assumption $f^2(p,z)-\beta(p)>0$, see \eqref{7.7}, the factor $\big(f^2(p,z_1)-\beta(p)\big)$ remains bounded when $p\to-\infty$. Also $\alpha(p)=\alpha_0e^{3p}\to0$ as $p\to-\infty$. Thus, the second term on the left-hand side of \eqref{8.5} runs to 0 as $p\to-\infty$. The same is true for the first term. Taking into account that $f^2(p,z_1)\to1$ as $p\to-\infty$, we obtain $\gamma(-\infty)=-\varepsilon$.

Since $\varepsilon=\pm1$ is the sign of $f'_p$ in \eqref{8.5}, by a similar analysis of the system \eqref{7.35}, we demonstrate that $\varepsilon=1$ in our setting. This proves the second equality in \eqref{8.4}.
The~first equality in \eqref{8.4} is proved similarly.

The ``Cauchy problem" \eqref{7.21}--\eqref{7.22}, \eqref{8.4} (with $\varepsilon=1$) is easily solved in series
$
\beta=\frac{1}{3}+\sum\nolimits_{k=1}^\infty\beta_k\alpha^k$ and
$
 \gamma=\frac{1}{3}+\sum\nolimits_{k=1}^\infty\gamma_k\alpha^k.
$
Equations \eqref{7.21}--\eqref{7.22} imply some recurrent relations that allow us to compute all coefficients.
In particular,
\begin{eqnarray}
\beta=\frac{1}{3}-\frac{7}{6}\alpha{+}\frac{13}{72}\alpha^2
{-}\frac{133}{1728}\alpha^3{+}\frac{575}{13824}\alpha^4
{-}\frac{2077}{82944}\alpha^5{+}\frac{37}{2304}\alpha^6+\dots,
                      \label{8.6} \\
\gamma=-1+\alpha-\frac{1}{8}\alpha^2
+\frac{7}{144}\alpha^3-\frac{115}{4608}\alpha^4
+\frac{67}{4608}\alpha^5-\frac{7}{768}\alpha^6+\dots.
                      \label{8.7}
\end{eqnarray}
{The series converge for all real $\alpha$. We omit the proof of the convergence which is not easy.}

The system \eqref{7.16} can be equivalently written in terms of the function $\psi(r,z)$. Indeed, as is seen from \eqref{8.2}, a solution $f(p,z)$ to the system \eqref{7.16} is related to $\psi$ by
$
\frac{1}{3}\ln\psi(f(p,z),z)+C=p.
$
Starting with this equation, we derive from \eqref{7.16} the system
\begin{equation}
\begin{aligned}
2c\psi'_r-3r(r^2+\gamma)&=0,\\
c^2\psi'_z{}^2-9c(r^2-\beta)\psi+\frac{9}{4}r^2(r^2+\gamma)^2&=0,
\end{aligned}
                                              \label{8.8}
\end{equation}
where $c$ is the constant from \eqref{8.3}.

By the change $\psi=\frac{1}{c}\tilde\psi$ of the unknown function, \eqref{8.8} is transformed to the same system with $c=1$.
Therefore we can assume $c=1$ without lost of generality, i.e.,
\begin{equation}
\begin{aligned}
2\psi'_r-3r(r^2+\gamma)&=0,\\
\psi'_z{}^2-9(r^2-\beta)\psi+\frac{9}{4}r^2(r^2+\gamma)^2&=0.
\end{aligned}
                                              \label{8.9}
\end{equation}
The equality \eqref{8.3} becomes
 $\alpha=\psi$.
The function $\psi(r,z)$ is defined and smooth in a neighborhood of the point $(1,0)$ and satisfies $\psi(1,0)=0$.
Let us show that $\psi(r,z)$ is an even function of $z$. Indeed, the second equation of the system \eqref{8.9} gives
$
 \psi'_z{}^2(1,0)+\frac{9}{4}r^2(1+\gamma(-\infty))^2=0.
$
Since $\gamma(-\infty)=-1$ by \eqref{8.4}, we obtain
\begin{equation}
\psi'_z(1,0)=0.
                                              \label{8.11}
\end{equation}
Then we differentiate the first equation of the system \eqref{8.9} with respect to $z$
\begin{equation}
 2\psi''_{rz}-3r\gamma'_z=0.
                                              \label{8.12}
\end{equation}
Since $\gamma$ depends on $p$ only,
$
\gamma'_z=\gamma' p'_z,
$
where $\gamma'=d\gamma/dp$. Together with \eqref{8.2}, this gives $\gamma'_z=\frac{\gamma'\psi'_z}{3\psi}$. Substituting this expression into \eqref{8.12} and setting $z=0$, we arrive to the linear first order ODE for the function $\psi'_z(r,0)$:
\begin{equation}
2\frac{d\psi'_z(r,0)}{d r}-3\frac{r\gamma'(p(r,0))}{3\psi(r,0)}\,\psi'_z(r,0)=0.
                                              \label{8.13}
\end{equation}
The coefficient $\frac{r\gamma'(p(r,0))}{3\psi(r,0)}$ { of the equation} is a bounded smooth function in a neighborhood of $r=1$ as is seen from \eqref{8.7}. Together with the initial condition \eqref{8.11}, the equation \eqref{8.13} implies
$\psi'_z(r,0)=0$.
The system \eqref{8.9} is invariant under the change $z\to-z$. Therefore the equality $\psi'_z(r,0)=0$ implies that $\psi(r,z)$ is an even function of $z$.

The system \eqref{8.9} allows us to compute term-by-term all Taylor coefficients of the function $\psi(r,z)$ at the point $(1,0)$.
In particular,
\begin{equation}
\begin{aligned}
\psi(r,z)&=\frac{3}{2}(r-1)^2+\frac{3}{2}z^2+\frac{9}{4}(r-1)^3+\frac{9}{4}(r-1)z^2+\frac{57}{32}(r-1)^4\\
&+\frac{45}{16}(r-1)^2z^2 +\frac{33}{32}z^4+\frac{9}{8}(r-1)^5+\frac{9}{4}(r-1)^3z^2+\frac{9}{4}(r-1)z^4+\dots
\end{aligned}
                              \label{8.14}
\end{equation}
Second and 3d order terms on the right-hand side of \eqref{8.14} are easily derived from \eqref{8.9},
and we used Maple for computing 4th and 5th order terms.

%
%
%

\begin{figure}
\noindent\includegraphics[scale=0.46]
{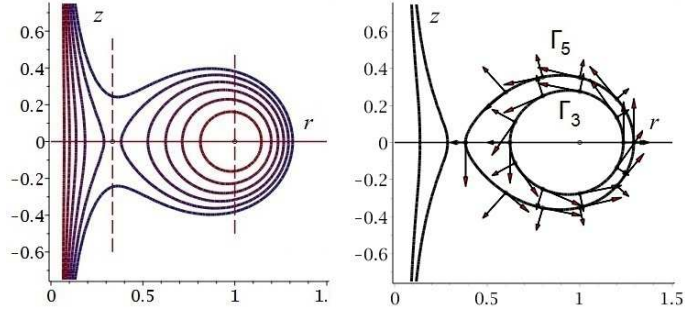}
\caption{\small
Structure of GF in a neighborhood of a minimum point of the pressure.
}
\label{fig:F2}
\end{figure}

{This GF is shown} on Figure~\ref{fig:F2}.
In our calculations we used the 5th order segment of the Taylor series of the function $\psi$,
i.e., we ignored  the remainder denoted by dots on the right-hand side of \eqref{8.14}.
Six isolines $\Gamma_i=\{(r,z): \psi(r,z)=\psi_i=0.04\,i\}\ (i=1,\dots,6)$ are drawn on the left-hand side of Figure~\ref{fig:F2}.
Each of the curves $\Gamma_1,\dots, {\Gamma}_5$ consists of two connected components while $\Gamma_6$ has one component.
The same curves $\Gamma_i$ are isolines of the pressure function, see \eqref{8.2}.
We set $C=-\frac{1}{3}\ln\psi_1$, thus the formula \eqref{8.2} becomes
$p(r,z)=\frac{1}{3}\ln\frac{\psi(r,z)}{\psi_1}$.
Each curve $\Gamma_i$ is the generatrix of the isobaric surface $M_{p_i}$, where $p_i=\frac{1}{3}\ln\frac{\psi_i}{\psi_1}$.
In our case, $p_i=\frac{1}{3}\ln i$.
The~velocity vector field $u$ is drawn on the right picture for $\Gamma_3$ and $\Gamma_5$.

{Observe the interesting phenomenon on Figure~\ref{fig:F2}: besides the minimum point $(1,0)$, the function $\psi(r,z)$ has the saddle point at $(r,z)=(1/3,0)$. The saddle point disappears when the 5th degree polynomial \eqref{8.14} is replaced with the corresponding 6th degree polynomial, and again appears at the same point for the 7th degree polynomial. We have no idea whether the phenomenon is essential for this type GFs, or it is just an artefact of ignoring higher degree terms in \eqref{8.14}.}

\section{Some open questions}
\label{sec:problems}

In our opinion, the main open question is: are there GFs on ${\mathbb R}^3$ which are not axisymmetric?
From an analytical point of view, this is a question on the consistency conditions {for} the system \eqref{5.12}--\eqref{5.15}.
The most expected answer to the question is ``yes''. However, the question is not easy because of the following.
For an axisymmetric GF, the constants $c$ and $C$ in the Bernoulli law \eqref{2.1} can be expressed through each other.
In other words, all particles living on an isobaric surface $M_p$ move with the same speed in the case of an axisymmetric GF.
Most likely, this statement is not true for a general (not axisymmetric) GF, at least we cannot prove it by local reasoning.
But maybe simple global arguments will help answer the question. For instance, if there exists a particle trajectory dense in $M_p$, then $\vert u\vert=\mbox{const}$ on $M_p$.

{We mostly} studied GFs locally in a neighborhood of a regular point.
The most interesting questions relate to GFs with regular compact isobaric hypersurfaces $M_p\subset{\mathbb R}^n$.
As mentioned in Introduction, application of the Gavrilov localization to such a flow gives a compactly supported GF on the whole of ${\mathbb R}^n$.
Since a compact regular isobaric hypersurface $M_p$ is endowed with a non-vanishing tangent vector field $u$, the Euler characteristic of $M_p$ is equal to zero. In~the most important 3D-case this means that $M_p$ is diffeomorphic to the two-dimensional torus~${\mathbb T}^2$.
The restriction of $u$ onto $M_p$ is a non-vanishing geodesic vector field. There are Riemannian metrics on ${\mathbb T}^2$ admitting a non-vanishing geodesic vector field, the corresponding example can be found in the class of so called double-twisted products \cite{Go}. But we are interested in metrics on ${\mathbb T}^2$ induced from the Euclidean metric of ${\mathbb R}^3$
by an embedding $i:{\mathbb T}^2\subset{\mathbb R}^3$. Apart from surfaces of revolution, we do not know any example related to the following~problem.

\begin{problem}
Classify triples $(i,u,\lambda)$, where $i:{\mathbb T}^2\subset{\mathbb R}^3$ is an embedding of the torus, $u$ is a non-vanishing geodesic vector field on ${\mathbb T}^2$ endowed with the Riemannian metric induced from the Euclidean metric of ${\mathbb R}^3$ by embedding $i$, and $\lambda>0$ is a smooth function on ${\mathbb T}^2$ satisfying the equations
\begin{equation}
\mbox{\rm div}\,u=u(\log\lambda),\quad II(u,u)=-\lambda,
                                              \label{10.1}
\end{equation}
where $II$ is the second quadratic form of $\,{\mathbb T}^2$.
\end{problem}

Equations \eqref{10.1} are obtained from \eqref{2.3} and \eqref{2.10} by setting $\lambda=\vert{\rm grad}\,p\,\vert$.

In this paper, we did not discuss the behavior of a GF near a critical point of the pressure. Such a discussion could be of great interest. For instance, it makes sense to study a GF $(u,p)$ with the Morse function $p$, i.e., all critical points of $p$ are non-degenerate. For such a flow, $M_p$ is still a regular hypersurface for a regular value of the pressure; but $M_p$ undergoes a Morse surgery when $p$ changes near a critical value of the pressure.
Which Morse surgeries are compatible with the Euler equations?


\end{document}